\documentclass[11pt]{article}

\usepackage[a4paper,margin=1in]{geometry}
\usepackage{amsthm}

\usepackage{amsmath,amsfonts}
\usepackage{bm}
\usepackage{graphicx}
\usepackage{subcaption}
\usepackage{mathtools}
\mathtoolsset{showonlyrefs}
\usepackage{xcolor}
\usepackage{hyperref}
\hypersetup{colorlinks=true,urlcolor=red,citecolor=black,linkcolor=black}
\usepackage{graphicx}
\usepackage[english]{babel}
\usepackage[utf8]{inputenc}
\usepackage{siunitx}
\usepackage{amsmath}
\usepackage{amsfonts}
\usepackage{float}
\usepackage{bm}
\usepackage{subcaption}
\usepackage{color,xcolor}
\usepackage[english]{babel}

\usepackage{todonotes}
\usepackage{mathtools}
\mathtoolsset{showonlyrefs}

\theoremstyle{plain}
\newtheorem{theorem}{Theorem}
\newtheorem{lemma}{Lemma}

\theoremstyle{definition}
\newtheorem{definition}{Definition}

\theoremstyle{remark}
\newtheorem{remark}{Remark}

\newcommand{\abs}[1]{\left\lvert#1\right\rvert}

\newcommand{\barpow}[1]{\left\lfloor#1\right\rceil}
\newcommand{\sign}[1]{\text{{\normalfont sign}}\left(#1\right)}
\newcommand{\graph}[1]{\mathcal{Q}\left(\mathcal{X}_{#1}\right)}
\newcommand{\graphm}[1]{\mathcal{M}(\hat{\mathcal{X}}_{#1})}

\def\Blue#1{#1}

\makeatletter
\renewcommand{\@fnsymbol}[1]{\@alph{#1}}
\makeatother

\begin{document}

\title{
Autonomous and non-autonomous fixed-time leader-follower consensus for second-order multi-agent systems%
\thanks{\textcolor{red}{
This is the accepted version of the manuscript: Trujillo, M.A., Aldana-López, R., Gómez-Gutiérrez, D. et al. Autonomous and non-autonomous fixed-time leader–follower consensus for second-order multi-agent systems. Nonlinear Dynamics 102, 2669–2686 (2020). DOI: 10.1007/s11071-020-06075-7. Please cite the publisher's version. For the publisher's version and full citation details see: \href{https://doi.org/10.1007/s11071-020-06075-7}{https://doi.org/10.1007/s11071-020-06075-7}.
}}
}

\author{
M. A. Trujillo\thanks{CINVESTAV Unidad Guadalajara, Av. del Bosque 1145, Zapopan, 45019, Jalisco, Mexico}
\and
R. Aldana-L\'opez\thanks{Department of Computer Science and Systems Engineering, University of Zaragoza, Zaragoza, 50009, Espa\~na}
\and
D. G\'omez Guti\'errez\thanks{Intel Tecnolog\'ia de M\'exico, Intel Labs, Multi agent autonomous systems lab, Jalisco, M\'exico. Tecnol\'ogico de Monterrey, Escuela de Ingenier\'ia y Ciencias, Jalisco, M\'exico. Email: David.Gomez.G@ieee.org}
\and
M. Defoort\thanks{LAMIH, UMR CNRS 8201, INSA, Polytechnic University of Hauts de France, Valenciennes, 59313 France}
\and
J. Ruiz Le\'on\footnotemark[1]
\and
H. M. Becerra\thanks{Centro de Investigaci\'on en Matem\'aticas (CIMAT), Jalisco S N, Col. Valenciana, 36023, Guanajuato, Mexico}
}
\date{}
\maketitle

\begin{abstract}
This paper addresses the problem of consensus tracking with fixed-time convergence, for leader-follower multi-agent systems with double-integrator dynamics, where only a subset of followers has access to the state of the leader. The control scheme is divided into two steps. The first one is dedicated to the estimation of the leader state by each follower in a distributed way and in a fixed-time. Then, based on the estimate of the leader state, each follower computes its control law to track the leader in a fixed-time. In this paper, two control strategies are investigated and compared to solve the two mentioned steps. The first one is an autonomous protocol which ensures a fixed-time convergence for the observer and for the controller parts where the Upper Bound of the Settling-Time (\textit{UBST}) is set a priory by the user. Then, the previous strategy is redesigned using time-varying gains to obtain a non-autonomous protocol. This enables to obtain less conservative estimates of the \textit{UBST} while guaranteeing that the time-varying gains remain bounded. Some numerical examples show the effectiveness of the proposed consensus protocols.
\end{abstract}

\section{Introduction}
In the last years, the problems of coordination and control of Multi-Agent System (MAS) have been widely studied (see for instance \cite{aldana2019predefined,li2014_Book,olfati2007,Parker2008,ren2008_Book}), due mainly to the ability of a MAS to face complex tasks that a single agent is not able to handle. Distributed control approaches applied to a MAS require a communication network allowing to share information with a subset of agents (neighbors). In this context, several interesting problems and applications have been investigated in the literature, for instance, synchronization of complex networks \cite{Arenas2008_sync}, distributed resource allocation \cite{Xu2017_ResAll}, consensus \cite{Olfati-Saber2004} and formation control of multiple agents \cite{Oh2015_Survey}. Among all the mentioned problems, an interesting one is the leader-follower consensus problem where a set of agents, through local interaction, converge to the state of a leader, even though the leader may not be accessible for all agents.

The consensus problem consists in reaching a common agreement state by exchanging only local information \cite{Olfati-Saber2004,Ren2007}. Linear average consensus protocols with asymptotic convergence were proposed in \cite{Olfati-Saber2004,Ren2007}. It has been demonstrated that the second smallest eigenvalue of the Laplacian graph (i.e. the algebraic connectivity) determines the convergence rate of the MAS. Furthermore, the problem of tracking a reference by a MAS (i.e. leader-follower consensus problem) has been investigated where the common agreement to reach is the state of a reference imposed by a leader which evolves independently of the MAS \cite{Cao2015,Guo2011,Ni2010,Ren2010}. In \cite{Ren2010}, the consensus problem has been addressed where the agents reach a time-varying reference. However, the control protocol has been derived for first-order MAS. The problem for second-order MAS has been studied in \cite{Guo2011} and extended to high-order MAS in \cite{Cao2015,Ni2010}. Furthermore, \cite{Lu2019} has considered the consensus tracking control problem of uncertain nonlinear MAS with predefined accuracy. Nevertheless, in these works, the convergence is only asymptotic.

To improve the convergence rate of a MAS, finite-time consensus protocols have been investigated in \cite{Shang2012}. Finite-time stability has been studied in \cite{Bhat2000,Moulay2006,Zhao2014}. However, the settling time is an unbounded function of the initial conditions of the system. Therefore, the concept of fixed-time stability has been introduced and applied to systems with time constraints \cite{Aldana2019,Polyakov2012,Sanchez-Torres2018}. In this case, the settling time is bounded by a constant which is independent of the initial conditions of the system. In the literature, there are several contributions on algorithms with fixed-time convergence property, such as stabilizing controllers \cite{Polyakov2012,Polyakov2015}, state observers \cite{Menard2017}, multi-agent coordination \cite{Aldana-Lopez2018a,Defoort2015},  online differentiation algorithms \cite{Angulo2013,Cruz-Zavala2011}, etc. Nevertheless, one can mention that the fixed-time stabilization problem of second-order systems is not an easy task since usually the settling time is not provided or is overestimated. Indeed, there are several works for second-order systems stabilization based on block-control techniques (\cite{Li2017,Polyakov2012,Zuo2015_iet,Zuo2016_CCC}) or on the homogeneity in the bi-limit (\cite{Tian2017}). However, the homogeneity-based algorithms do not provide an estimate of the settling time and many block-control-based algorithms neglect some transient when the system trajectories stay on a region around a manifold.  Moreover, the works \cite{Khoo2009,Khoo2014,Zhang2012_Auto,Zuo2015_b} deal with the problem of leader-follower consensus. Nevertheless, these algorithms require that each follower know the inputs of its neighbors simultaneously, which causes communication loop problems. In this paper, we address the leader-follower consensus problem of a MAS, where each agent of the MAS estimates and tracks the trajectory of the leader using local available information even when just a subset of MAS has access to the leader state, and we provide the necessary conditions to achieve the convergence in a fixed-time.

A Lyapunov differential inequality for an autonomous system to exhibit fixed-time stability was presented in \cite{Polyakov2012}. Based on this methodology using autonomous systems, the consensus problem with fixed-time convergence property has been derived for first-order MAS in \cite{Aldana-Lopez2018a,Zuo2014,Zuo2016}. Nevertheless, in \cite{Zuo2016}, the \textit{UBST} has been estimated from design parameters, algebraic connectivity and group order. Thus, it cannot be easily tuned. In \cite{Aldana-Lopez2018a}, the \textit{UBST} was a design parameter which was established a priory by the user. However, the settling time becomes over-estimated and the slack between the settling time and the \textit{UBST} is  conservative. Furthermore, the works \cite{Anggraeni2019,Ning2018,Zuo2018} have addressed the consensus tracking problem, i.e., the MAS follows a trajectory imposed by the leader. The scheme presented in \cite{Ning2018} has introduced a fixed-time algorithm considering inherent dynamics for the agents. However, disturbances were not taken into account. The leader-follower consensus problem for agents with second-order and high order integrator dynamics has been addressed in~\cite{Tian2019,Zuo2018}, respectively. The approach was based on a fixed-time observer to estimate the leader state and a fixed-time controller to drive the state of the agent to the estimated leader state. Unfortunately, although the observer can be designed to converge at a desired \textit{UBST} (with a conservative  estimate of the \textit{UBST}), the controller is based on the homogeneity theory~\cite{Andrieu2008} and no methodology has been provided to estimate an \textit{UBST}. Thus, although the algorithm is fixed-time convergent, the desired convergence time cannot be set a priory by the user. To address this issue, autonomous algorithms were proposed in~\cite{Ni2017_a,Ni2019,Shi2020} with an estimation of the \textit{UBST}. Unfortunately, such estimate of the \textit{UBST} results very conservative leading to over-engineered consensus protocols. Therefore, the design of fixed-time leader-follower consensus algorithms where the \textit{UBST} is set explicitly as a parameter of the system, as well as the reduction of the conservativeness of the estimate of the \textit{UBST} is of a great interest. 

An approach to derive predefined-time consensus algorithms has been addressed via a linear function of the sum of the errors between neighboring nodes together with a time-varying gain, using time base generators~\cite{Morasso1997}, see e.g.,~\cite{Colunga2018c,Liu2018,Ning2019,Wang2017,Wang2018,Yong2012,Yucelen2019,Zhao2018}. This approach ensures that the convergence is obtained exactly at a predefined time. However, such time-varying gain becomes singular at the predefined time, either because the gain goes to infinite as the time tends to the predefined time~\cite{Yong2012,Yucelen2019,Zhao2018} or because it produces Zeno behavior (infinite number of switching in a finite-time interval) as the time tends to the predefined time~\cite{Liu2018}. 

In this paper, we present a methodology to achieve leader-follower consensus with fixed-time convergence. It consists in two steps. The first one estimates the leader state (position and velocity) using a fixed-time observer that only requires information of the neighbors. Then, the second step computes the control law to drive the followers to the observer states in a fixed-time. Moreover, we investigate two protocols to solve the considered problems. In one protocol, called autonomous protocol, the convergence for the observer and for the controller is in fixed-time, where the \textit{UBST} is established a priory by the user. In the second protocol, called non-autonomous protocol, we redesign the previous one by adding time-varying gains to obtain a less conservative estimate of the \textit{UBST} while guaranteeing that the time-varying gains remain bounded. The contribution lies in the following. A novel protocol is derived for second-order MAS with fixed-time stability where the \textit{UBST} is a design parameter. Moreover, a non-autonomous protocol is presented to achieve the convergence in a predefined-time with less conservative estimates of the \textit{UBST} compared to existing results in the literature. In fact, the resulting \textit{UBST} can be made arbitrarily tight. At last, our algorithm yields a bounded time-varying gain, thus we avoid the drawbacks present in the existing algorithms with time-varying gains where the gain goes to infinity.

This work is structured as follows. Section \ref{sec:preliminaries} recalls some definitions and results from graph theory, and preliminaries on finite-time and fixed-time convergence are presented. In Section \ref{sec:problem_statement}, the problem of consensus tracking with fixed-time convergence is formulated. Section \ref{Sec:AutConsensus} introduces two methodologies to solve the consensus tracking problem. The first (resp.  second) one is based on algorithms to obtain a fixed-time stable autonomous (resp. non autonomous) system with an \textit{UBST} function independent of the initial conditions of the system. Numerical results using both methodologies are shown in Section \ref{sec:results}. Finally, the conclusions are presented in Section \ref{sec:conclusions}.
 
\section{Preliminaries}\label{sec:preliminaries}
\subsection{Graph Theory}

In this section, some notations and preliminaries about graph and consensus theory  are presented. One can refer to \cite{Godsil2001,Olfati-Saber2007} for a deeper insight in these fields. This paper is only focused on undirected graphs for the follower agents.

\begin{definition}
A graph consists of a set of vertices $\mathcal{V}(\mathcal{X})$ and a set of agents $\mathcal{E}(\mathcal{X})$ where an edge is an unordered pair of distinct vertices of $\mathcal{X}$. $ij$ denotes an edge, if vertex $i$ and vertex $j$ are adjacent or neighbors. The set of neighbors of $i$ in graph $\mathcal{X}$ is expressed by $\mathcal{N}_{i}(\mathcal{X})=\{j\in \mathcal{X} :ji\in\mathcal{E}(\mathcal{X})\}$.
\end{definition}

\begin{definition}
A path from $i$ to $j$ in a graph is a sequence of distinct vertices starting with $i$ and ending with $j$ such that consecutive vertices are adjacent. If there is a path between any two vertices of graph $\mathcal{X}$, then $\mathcal{X}$ is said to be connected. 
\end{definition}

\begin{definition}
Let $\mathcal{X}$ be a weighted graph such that $ij\in \mathcal{E}$ has weight $a_{ij}$ and let $N=|\mathcal{V}(\mathcal{X})|$. Then, the adjacency matrix $A(\mathcal{X})$  (or simply $A$ when the graph is clear from the context) is an $N\times N$ matrix where $A=[a_{ij}]$ and, the Laplacian is denoted by $\mathcal{Q}(\mathcal{X})$ (or simply $\mathcal{Q}$) and is defined as $\mathcal{Q}(\mathcal{X})=\Delta(\mathcal{X})-A(\mathcal{X})$ where $\Delta(\mathcal{X})=\text{diag}(d_{1},...,d_{N})$ with $d_{i}=\sum_{j\in \mathcal{N}_{i}}a_{ij}$.   
\end{definition}
Through this work, it is assumed that $a_{ij}=a_{ji}$, i.e. only undirected and balanced graphs are considered.


\begin{definition} 
Let $\hat{\mathcal{X}}$ be a weighted graph among all the agents (i.e. the leader and the followers). Then, the communication matrix between all the agents is represented by $\graphm{}=\graph{}+\mathcal{B}$ where $\mathcal{B}=\text{diag}(b_{1}\hdots b_{N}) \in \mathbb{R}^{N\times N}$  with $b_{i}>0$ when there is an edge from the leader to the $i$-agent and $\graph{}$ is the weighted graph associated to the communication topology of the followers.
\end{definition}
\begin{lemma}\cite{Ren2013,Ni2010} 
Let $\hat{\mathcal{X}}$ be the communication graph among all the agents with the leader as the root. Then, matrix $\graphm{}$ is symmetric positive definite. 
\end{lemma}
\subsection{On finite-time and fixed-time}

Consider the system
\begin{equation}\label{eq:x_function}
    \dot{x}(t)=f\left(x(t),t;\rho\right), \ \ \ \ x(0)=x_0
\end{equation}
where $x\in\mathbb{R}^{n}$ is the system state, the vector  $\rho \in \mathbb{R}^{b}$ stands for more parameters of system \eqref{eq:x_function} which are assumed to be constant, i.e., $\dot{\rho}=0$. Furthermore, there is no limit for the number of parameters, so $b$ can take any value in the natural number set $\mathbb{N}$. The function $f:\mathbb{R}^{n}\times \mathbb{R}_+\rightarrow \mathbb{R}^{n}$ is nonlinear and the origin is assumed to be an equilibrium point of system \eqref{eq:x_function}, so that $f(0,t;\rho)=0$. Besides, when function $f$ does not depend explicitly on $t$, the system is said to be autonomous or time-invariant. Otherwise, it is called non-autonomous or time-varying \cite{Khalil2002}.
\begin{definition}
\cite{Bhat2005} The origin of \eqref{eq:x_function} is globally finite-time stable if it is globally asymptotically stable and any solution $x(t;x_{0})$ of \eqref{eq:x_function} reaches the equilibrium point at some finite time moment i.e. $x(t;x_{0})=0, \forall t \geq T(x_{0})$ where $T:\mathbb{R}^{n}\rightarrow\mathbb{R}_{+}\cup \{0\}$ is called the settling-time function.
\end{definition}

\begin{definition}
\cite{Polyakov2016} The origin of \eqref{eq:x_function} is fixed-time stable if it is globally finite-time stable and the settling function is bounded, i.e., $\exists T_{\max}>0: T(x_{0})\leq T_{\max}, \forall x_{0} \in \mathbb{R}^{n}$.
\end{definition}

\begin{theorem}\cite{Aldana-Lopez2018} \label{th:cons_pred}
Consider the system 
\begin{equation}\label{eq:sign_pqk}
    \dot{x}=-(\alpha \abs{x}^{p}+\beta \abs{x}^{q})^{k}\sign{x},\ \ \ 
    x(0)=x_{0}
\end{equation}
with $x\in \mathbb{R}$. The parameters of the system are the real numbers $\alpha,\beta,p,q,k>0$ which satisfy the constraints $kp<1$ and $kq>1$. Let $\rho = [\alpha, \beta, p, q, k]^{T} \in \mathbb{R}^{5}$. Then, the origin $x=0$ of system  \eqref{eq:sign_pqk} is fixed-time stable and  the settling time function satisfies $T(x_0)\leq T_{f}=\gamma(\rho)$, where 
\begin{equation}\label{eq:gamma_rho}
    \gamma(\rho)=\frac{\Gamma\left(\frac{1-kp}{q-p}\right)\Gamma\left(\frac{qk-1}{q-p}\right)}{\alpha^{k}\Gamma(k)(q-p)}\left(\frac{\alpha}{\beta} \right)^{\frac{1-kp}{q-p}},
\end{equation}
and $\Gamma(\cdot)$ is the Gamma function defined as $\Gamma (z)= \int^{\infty}_{0}e^{-t}t^{z-1}dt$ (see \cite{Bateman1953} for details on the Gamma function).
\end{theorem}

\begin{theorem}\cite{Aldana-Lopez2018} \label{thm:pred_time_stable}
Consider the system 
\begin{equation}\label{eq:system}
    \dot{x}(t)=f(x(t),t;\rho), \ \ \
    x(0)=x_{0}
\end{equation}
where $x\in \mathbb{R}^{n}$ is the system state, the vector $\rho \in \mathbb{R}^{b}$ stands for the system parameters which are assumed to be constant. The function $f:\mathbb{R}^{n}\times \mathbb{R}_{+}\rightarrow\mathbb{R}^n$ is such that $f(0,t;\rho)=0$. Assume that there exists a continuous radially unbounded function $V:\mathbb{R}^{n} \rightarrow \mathbb{R}$ such that:
\begin{align}
 V(0)&=0 \\
 V(x)&>0, \quad \forall x \in \mathbb{R}^{n} \backslash \{0\}
\end{align} 
and the derivative of $V$ along the trajectories of \eqref{eq:system} satisfies 
\begin{equation}
    \dot{V}(x)\leq-\frac{\gamma(\rho)}{T_{c}}(\alpha V(x)^{p} + \beta V(x)^{q})^{k}, \forall x \in \mathbb{R}^{n} \backslash \{0\}
\end{equation} 
where $\alpha,\beta,p,q,k >0$, $kp<1$, $kq>1$, $\gamma$ is given in \eqref{eq:gamma_rho} and $\dot{V}$ is the upper right-hand time-derivative of $V$. Then, the origin of \eqref{eq:system} is predefined-time stable with predefined-time $T_{c}$.
\end{theorem}
\begin{definition}\label{def:}
For any  real number $r$, the function  $x\mapsto\barpow{x}^{r}$ is defined as $\lfloor x \rceil^{r}= |x|^{r}\sign{x}$ for any $x\in \mathbb{R}$ if $r>0$, and for any $x\in\mathbb{R} \setminus 0$ if $r\leq0$. Moreover, if $r>0$,$\barpow{0}^{r}=0.$
\end{definition}

\section{Problem statement}\label{sec:problem_statement}

Let us consider a group of $N+1$ agents with one leader and $N$ followers labeled $0$ and $i \in \{1,\ldots,N\}$, respectively. The dynamics of the leader is described by
\begin{equation}
    \begin{array}{lll}
    \dot x_{0}(t)&=&v_0(t) \\
    \dot v_0(t)&=&u_0(t)
    \end{array}
\end{equation}
where $X_0=[x_0,v_0]^T\in\mathbb{R}^2$ is the state of the leader and $u_0\in\mathbb{R}$ is the control input of the leader, which is assumed to satisfy $|u_0(t)|\leq u_0^{max}, \forall t\geq 0$ with $u_0^{max}$ a known constant. The dynamics of the $i-$th follower agent is given by:
\begin{align}
       \dot x_{i}(t)&=v_i(t) \\
    \dot v_i(t)&=u_i(t)+\Delta_i(t)\label{eq:din_agent_i} 
\end{align}
where $X_i=[x_i,v_i]^T\in\mathbb{R}^2$ is the state of agent $i$, $u_i\in\mathbb{R}$ is the control input of agent $i$ and $\Delta_i$ is an unknown external disturbance which is assumed to satisfy $|\Delta_i(t)|\leq \delta _i, \forall t\geq 0$ with $\delta _i$ a known constant. Besides, each agent estimates the leader states, represented by $\hat{x}_{i}$ (position) and $\hat{v}_{i}$ (velocity). The communication topology is represented by an undirected graph, which is assumed to contain a spanning tree with the leader agent as the root. The $i-$th agent shares its estimated states of the leader with its neighbors, defined by the neighbor set $\mathcal{N}_i$.

The control objective is to design a distributed control $u_i$ such that the consensus is achieved in a fixed-time $T_{c}$, i.e.
\begin{equation}
    \left\{
    \begin{array}{l}
        \lim_{t\rightarrow T_c}\|X_i(t)-X_0(t)\| =0 \\
        X_i(t) = X_0(t), \qquad \forall t>T_{c}.
        \end{array}\right.
\end{equation}

This goal is achieved into two stages. An ``observer", based on consensus algorithms, allows each agent to obtain an estimate of the leader state in a distributed manner in a fixed-time. Then, after the observer converges, a controller drives the state of the agent towards the state trajectory of the leader. Two protocols are investigated hereafter. In the first one, known as an autonomous protocol,  we guarantee that each agent is driven towards the leader state in a fixed-time, where the Upper Bound of the Settling-Time (UBST) is specified a priory by the user. In the second one, known as a non-autonomous protocol, we redesign the previous one by adding time-varying gains to obtain a less conservative estimate of the \textit{UBST} while guaranteeing that the time-varying gains remain bounded. 

\section{Fixed-time leader-follower consensus using autonomous protocols}
\label{Sec:AutConsensus}

\subsection{Distributed fixed-time observer}

Since the leader state is not available to all followers, for each agent, an observer is designed to estimate the state of the leader in a fixed-time. The observer has the following structure:
\begin{equation} \label{eq:observer_dynamics}
\begin{array}{lll}
    \dot{\hat{x}}_{i}&=&\hat{v}_i - \kappa_{i,x}\left[(\alpha|e_{1,i}|^{p} + \beta|e_{1,i}|^{q})^k+\zeta_x\right] \sign{e_{1,i}}\\
    \dot{\hat{v}}_i &=& -\kappa_{i,v}\left[(\alpha|e_{2,i}|^{p} + \beta|e_{2,i}|^{q})^k+\zeta_v\right] \sign{e_{2,i}}
    \end{array}
\end{equation}
with $ e_{1,i}=\sum_{j\in\mathcal{N}_i} a_{ij} (\hat{x}_j(t)-\hat{x}_i(t))+b_i(x_0(t)-\hat{x}_i(t))$ and $ e_{2,i}=\sum_{j\in\mathcal{N}_i} a_{ij} (\hat{v}_j(t)-\hat{v}_i(t))+b_i(v_0(t)-\hat{v}_i(t))$, $\hat{x}_{i}$ (resp. $\hat{v}_{i}$) is the estimate of the leader position (resp. velocity) for the $i$-th follower. $\kappa_{i,x}$, $\kappa_{i,v}$, $\alpha$, $\beta$, $k$, $p$, $q$, $\zeta_x$ and $\zeta_v$ are positive constants to be defined later. 
For each agent, let us denote the observer errors as
\begin{equation}
    \begin{array}{lll}
    \tilde x_i&=& \hat{x}_i - x_0\\
    \tilde v_i&=& \hat{v}_i - v_0.
    \end{array}\label{eq:error_obs_leader}
\end{equation}
Therefore, the observation error dynamics can be expressed as:
\begin{equation}
\begin{array}{lll}
    \dot{\tilde{x}}_{i}&=&\tilde{v}_i - \kappa_{i,x}\left[(\alpha|e_{1,i}|^{p} + \beta|e_{1,i}|^{q})^k+\zeta_{x}\right] \sign{e_{1,i}}\\
    \dot{\tilde{v}}_i &=& -\kappa_{i,v}\left[(\alpha|e_{2,i}|^{p} + \beta|e_{2,i}|^{q})^k+\zeta_{v}\right] \sign{e_{2,i}}-u_{0}
    \end{array}\label{eq:x_tilde}
\end{equation}
with $ e_{1,i}=\sum_{j\in\mathcal{N}_i} a_{ij} (\tilde{x}_j(t)-\tilde{x}_i(t))-b_i \tilde{x}_i(t))$ and $ e_{2,i}=\sum_{j\in\mathcal{N}_i} a_{ij} (\tilde{v}_j(t)-\tilde{v}_i(t))-b_i\tilde{v}_i(t))$.

In a compact form, with $\tilde{x}=[\tilde{x}_1 \ \cdots \ \tilde{x}_N]^T\in\mathbb{R}^N$ and $\tilde{v}=[\tilde{v}_1 \ \cdots \ \tilde{v}_N]^T\in\mathbb{R}^N$,
system \eqref{eq:x_tilde} can be written as:
\begin{equation}\label{eq:compact_form}
    \begin{array}{lll}
    \dot{\tilde{x}} &=& \tilde{v} -\Phi_x\left(\graphm{}\tilde{x}\right)\\
    \dot{\tilde{v}} &=& -\Phi_v\left(\graphm{}\tilde{v}\right)- \boldsymbol{1} u_0
    \end{array}
\end{equation}
where $\graphm{}$ represents the connection matrix of the graph describing the network between the followers and the leader, and for $z=[z_1 \ \cdots \ z_N]^T\in\mathbb{R}^N$, the functions $\Phi_x:\mathbb{R}^N\rightarrow \mathbb{R}^N$ and $\Phi_v:\mathbb{R}^N\rightarrow \mathbb{R}^N$ are defined as
\begin{align}
\Phi_x(z)=
\left[
\begin{array}{c}
\kappa_{1,x}\left[(\alpha|z_1|^{p} + \beta|z_1|^{q})^k+\zeta_x\right] \sign{z_1} \\ \vdots \\ \kappa_{N,x}\left[(\alpha|z_N|^{p} + \beta|z_N|^{q})^k+\zeta_x\right] \sign{z_N}
\end{array}
\right],\\
\Phi_{v}(z)=
\left[
\begin{array}{c}
\kappa_{1,v}\left[(\alpha|z_1|^{p} + \beta|z_1|^{q})^k+\zeta_v\right] \sign{z_1} \\ \vdots \\ \kappa_{N,v}\left[(\alpha|z_N|^{p} + \beta|z_N|^{q})^k+\zeta_v\right] \sign{z_N}
\end{array}
\right].
\end{align}

\begin{theorem}\label{thm:predefined_time_auto}
If the observer parameters are selected as $\alpha, \beta, p, q, k > 0$, $kp<1$, $kq>1$, $\zeta_x \geq 0$,
\begin{align}
\kappa_{x}\geq \frac{N\gamma(\rho)}{\lambda_{\min}(\graphm{})T_{c_{2}}}\text{, }
    \kappa_{v}\geq \frac{N\gamma(\rho)}{\lambda_{\min}(\graphm{})T_{c_{1}}}\text{ and } \kappa_{v}\zeta_{v}\geq u_0^{max}
\end{align}
where $$\kappa_{x}= \min_{i\in\{1\dots N\}} \kappa_{i,x} \text{ and } \kappa_{v}= \min_{i\in \{1\dots N\}} \kappa_{i,v}$$ and $\gamma(\rho)$ is defined in Equation \eqref{eq:gamma_rho}, then under the distributed observer \eqref{eq:observer_dynamics}, the observer error dynamics \eqref{eq:x_tilde} is fixed-time stable with a predefined-time $T_{o}=T_{c_1}+T_{c_2}$.
\end{theorem}

\begin{proof}
Consider the radially unbounded Lyapunov function candidate
\begin{equation}
 V_1(\tilde{v})=\frac{1}{N}\sqrt{\lambda_{\min}\left( \graphm{} \right) \tilde{v}^{T}\graphm{} \tilde{v} }.\label{eq:lyapunov_proposed}   
\end{equation}

Its time-derivative along the trajectories of system \eqref{eq:compact_form} is
\begin{equation}\label{eq:lyapunov_derivative}
\dot V_1=\frac{\sqrt{\lambda_{\min}\left(\graphm{}\right)}}{N}\frac{\tilde{v}^{T}\graphm{} \dot{\tilde{v}}}{\sqrt{\tilde{v}^{T}\graphm{} \tilde{v}}}
\end{equation}

Let us denote $e_{2}=\graphm{}\tilde{v}=[e_{2,1}\ \cdots \ e_{2,N}]^T$. Then, Equation \eqref{eq:lyapunov_derivative} can be written as follows
\begin{align} 
\dot{V}_1&=\frac{\sqrt{\lambda_{\min}\left(\graphm{}\right)}}{N}\frac{e_{2}^{T}}{\sqrt{\tilde{v}^{T}\graphm{} \tilde{v}}}\left( -\Phi_v\left(\graphm{}\tilde{v}\right)-\boldsymbol{1}u_{0}\right)\\
&
\begin{multlined}
=\frac{\sqrt{\lambda_{\min}\left(\graphm{}\right)}}{N}\left(-\frac{1}{\sqrt{\tilde{v}^{T}\graphm{} \tilde{v}}} \sum_{i=1}^{N}\kappa_{i,v}e_{2,i}\left[(\alpha|e_{2,i}|^{p} + \beta|e_{2,i}|^{q})^k+\right.\right.\\\left.\left.\cdots+\zeta_{v}\right]\sign{e_{2,i}}-\frac{e_{2}^{T}\boldsymbol{1}u_{0}}{\sqrt{\tilde{v}^{T}\graphm{} \tilde{v}}}\right)
\end{multlined}
\\
&
\begin{multlined}
=\frac{\sqrt{\lambda_{\min}\left(\graphm{}\right)}}{N}\left(-\frac{1}{\sqrt{\tilde{v}^{T}\graphm{} \tilde{v}}} \sum_{i=1}^{N}\kappa_{i,v}\abs{e_{2,i}}(\alpha|e_{2,i}|^{p} + \beta|e_{2,i}|^{q})^{k}\right. \\\left.-\frac{\zeta_{v}}{\sqrt{\tilde{v}^{T}\graphm{} \tilde{v}}}\sum_{i=1}^{N}\kappa_{i,v}\abs{e_{2,i}}-\frac{e_{2}^{T} \boldsymbol{1}u_{0}}{\sqrt{\tilde{v}^{T}\graphm{} \tilde{v}}}\right)\label{eq:lyap_deni_der}
\end{multlined}
\end{align}

Now, using the inequality \eqref{eq:lemma_sum_expo} of Lemma \ref{lem:sum_expo} in Appendix, the first term of Equation \eqref{eq:lyap_deni_der} can be written as 
\begin{multline}
 \frac{1}{\sqrt{\tilde{v}^{T}\graphm{} \tilde{v}}} \sum_{i=1}^{N}\kappa_{i,v}\abs{e_{2,i}}(\alpha|e_{2,i}|^{p} + \beta|e_{2,i}|^{q})^{k} 
 \geq\\\frac{N\kappa_{v}}{\sqrt{\tilde{v}^{T}\graphm{} \tilde{v}}}\left(\frac{1}{N}\sum_{i=1}^{n}\abs{e_{2,i}}\right)\left(\alpha \left(\frac{1}{N}\sum_{i=1}^{N}\abs{e_{2,i}}\right)^{p}+ \beta \left(\frac{1}{N}\sum_{i=1}^{N}\abs{e_{2,i}}\right) ^{q} \right)^{k}
 \end{multline}
with $\kappa_{v}=\min\{\kappa_{1,v},\dots,\kappa_{N,v}\}$. Since $\| e_{2}\|_{1}= \sum_{i}^{N} \abs{e_{2,i}}$ the last expression can be written as 
\begin{multline}
 \frac{1}{\sqrt{\tilde{v}^{T}\graphm{} \tilde{v}}} \sum_{i=1}^{N}\kappa_{i,v}\abs{e_{2,i}}(\alpha|e_{2,i}|^{p} + \beta|e_{2,i}|^{q})^{k} \geq\\ \frac{N\kappa_{v}}{\sqrt{\tilde{v}^{T}\graphm{} \tilde{v}}} \left(\frac{1}{N} \|e_{2}\|_1\right) \left(\alpha \left(\frac{1}{N} \| e_{2}\|_1\right)^{p}+ \beta \left(\frac{1}{N}  \| e_{2}\|_1 \right) ^{q} \right)^{k}.
\end{multline}
Furthermore, from Lemma \ref{lem:sum_abs} in Appendix, one gets
\begin{align}
 \|e_{2}\|_{1}\geq \| e_{2} \|_{2} = \sqrt{e_{2}^{T}e_{2}}= \sqrt{\tilde{v}^{T}\graphm{}^{2}\tilde{v}} \geq \sqrt{\lambda_{\min}\left(\graphm{}\right)\tilde{v}^{T}\graphm{}\tilde{v}}.
\end{align}{}
Hence, 
\begin{multline*}
    \frac{1}{\sqrt{\tilde{v}^{T}\graphm{} \tilde{v}}} \sum_{i=1}^{N}\kappa_{i,v}\abs{e_{2,i}}(\alpha|e_{2,i}|^{p} + \beta|e_{2,i}|^{q})^{k} 
 \geq \\ \frac{\kappa_{v} N}{\sqrt{\tilde{v}^{T}\graphm{} \tilde{v}}}  V  \left(\alpha  V^{p}+ \beta V  ^{q} \right)^{k}=\kappa_{v}\sqrt{\lambda_{\min}\left(\graphm{}\right)} \left(\alpha V_1^{p}+ \beta V_1  ^{q} \right)^{k}.\label{eq:lyap_der_term_1} 
\end{multline*}

Now, for the last two terms of Equation \eqref{eq:lyap_deni_der}, one can obtain
\begin{multline*}
 -\frac{\zeta_{v}}{\sqrt{\tilde{v}^{T}\graphm{} \tilde{v}}}\sum_{i=1}^{N}k_{i,v}\abs{e_{2,i}}-\frac{e_{2}^{T} \boldsymbol{1}u_{0}}{\sqrt{\tilde{v}^{T}\graphm{} \tilde{v}}}
\leq \\ -\frac{\|e_{2}\|_1}{\sqrt{\tilde{v}^{T}\graphm{}\tilde{v}}} \left(\kappa_{v} \zeta_{v}-u_0^{max} \right)\leq 0.
\end{multline*}
Therefore, the following inequality can be obtained
\begin{equation}
    \dot V_1\leq- \frac{\kappa_{v}\lambda_{\min}\left(\graphm{}\right)}{N} \left(\alpha V_1^{p}+ \beta V_1  ^{q} \right)^{k}.
\end{equation}
From Theorem \ref{thm:pred_time_stable}, the observation error in velocity $\tilde{v}$ converges to the origin in a fixed-time before the predefined-time $T_{c_1}$ where $\gamma(\rho)$ is given by Eq. \eqref{eq:gamma_rho}.

Once the observation error in velocity $\tilde{v}$ converges to zero (i.e. after time $T_{c_1}$), the observation error dynamics in position can be reduced to
\begin{equation}
    \dot{\tilde{x}}_{i}=-\kappa_{i,x}\left[(\alpha|e_{1,i}|^{p} + \beta|e_{1,i}|^{q})^k+\zeta_x\right]\sign{e_{1,i}}.
\end{equation}{}
Similarly to the previous analysis, one can easily show that 
\begin{equation}
 V_2(\tilde{x})=\frac{1}{N}\sqrt{\lambda_{\min}\left( \graphm{} \right) \tilde{x}^{T} \graphm{}\tilde{x} }\label{eq:lyapunov_proposed2}   
\end{equation}
satisfies
\begin{equation}
\dot V_2(\tilde{x})\leq- \frac{\kappa_{x}\lambda_{\min}\left(\graphm{}\right)}{N} \left(\alpha V_2^{p}+ \beta V_2  ^{q} \right)^{k}
 \label{equ:lyap_x_observer}, \forall t\geq T_{c_2}.
\end{equation}

From Theorem \ref{thm:pred_time_stable}, the observation error in position $\tilde{x}$ converges to the origin in a fixed-time before the predefined-time $T_{c_2}$.

Therefore, the proposed distributed observer guarantees the estimation of the leader states in a fixed-time before the predefined-time $T_{o}=T_{c_{1}}+T_{c_{2}}$.
\end{proof}

\subsection{A Fixed-time tracking controller}
After time $T_o$, each agent has an accurate estimation of the leader state. For each agent, the tracking error is defined as
\begin{equation}\label{eq:errors2}
    \begin{array}{lll}
    e_{x,i}&=&x_i - \hat{x}_i \\
    e_{v,i}&=&v_i - \hat{v}_i,
    \end{array}
\end{equation}
or, equivalently, after the convergence of the observation error:
\begin{equation}\label{eq:errors}
    \begin{array}{lll}
    e_{x,i}&=&x_i - x_0 \\
    e_{v,i}&=&v_i - v_0.
    \end{array}
\end{equation}

Its dynamics can be expressed as:
\begin{equation}\label{eq:sys_2dn_order}
\begin{array}{lll}
    \dot e_{x,i} &=& e_{v,i}\\
    \dot e_{v,i} &=& u_i+\Delta_i - u_0.
    \end{array}
\end{equation}

Here, the objective is to design the control input $u_i$ such that the origin $(e_{x,i},e_{v,i})=(0,0)$ is fixed-time stable where the Upper Bound of the Settling-Time (\textit{UBST}) is set a priory by the user, in spite of the unknown but bounded perturbation term $\Delta_i-u_0$. Herefater, we present the following results motivated by the work \cite{Aldana-Lopez2018}.

\begin{theorem}\label{thm:predefined_time_cont_auto}
If for each agent, the controller is selected as 
\begin{equation}\label{eq:uso}
u_i=\upsilon(e_{x,i},e_{v,i})=-\left[\frac{\gamma_2}{\hat{T}_{c_2}}\left(\alpha_2\abs{\sigma_i}^{p'}+\beta_2\abs{\sigma_i}^{q'}\right)^{k'}+\frac{\gamma_1^2}{2\hat{T}_{c_1}^2}\left(\alpha_1+3\beta_1e_{x,i}^2\right)+\zeta_i(t)\right]\sign{\sigma_i}
\end{equation}
with the following sliding variable
\begin{equation}\label{eq:sigmaso}
\sigma_i=e_{v,i}+\barpow{\barpow{e_{v,i}}^2+\frac{\gamma_1^2}{\hat{T}_{c_1}^2}\left(\alpha_1\barpow{e_{x,i}}^1+\beta_1\barpow{e_{x,i}}^3\right)}^{1/2},
\end{equation}
where parameters are selected as $\alpha_1,\alpha_2,\beta_1,\beta_2,T_o',\hat{T}_{c_1},\hat{T}_{c_2}>0$, $p', q', k' > 0$, $k'p'<1$, $k'q'>1$, $\zeta_i \geq u_0^{max}+\delta_i$, $\gamma_1=\frac{\Gamma \left(\frac{1}{4}\right)^2 }{2\alpha_1^{1/2}\Gamma\left(\frac{1}{2}\right)}\left(\frac{\alpha_1}{\beta_1}\right)^{1/4}$ and $\gamma_2 = \frac{\Gamma(m_p)\Gamma(m_q)}{\alpha_2^{k'}\Gamma(k')(q'-p')}\left(\frac{\alpha_2}{\beta_2}\right)^{m_p}$ with $m_p=\frac{1-k'p'}{q'-p'}$ and $m_q=\frac{k'q'-1}{q'-p'}$, then the leader-follower consensus is achieved in a predefined-time $\hat{T}_c=
T'_{o}+\hat{T}_{c_1}+\hat{T}_{c_2}$.
\end{theorem}

\begin{proof}
First, the time derivative of  $\sigma_i$ along the trajectory of the system solution is given by
\begin{align}
\dot{\sigma}_{i}=&u_{i}+\Delta_i-u_0+ \frac{\abs{e_{v,i}}({u}_{i}+\Delta_i-u_0)+\frac{\gamma_1^2}{2\hat{T}_{c_1}^2}\left(\alpha_1+3\beta_1e_{x,i}^2\right)e_{v,i}}{\abs{\barpow{e_{v,i}}^2+\frac{\gamma_1^2}{\hat{T}_{c_1}^2}\left(\alpha_1\barpow{e_{x,i}}^1+\beta_1\barpow{e_{x,i}}^3\right)}^{1/2}}.\\
\end{align}{}
Using the control input $u_{i}$ given by \eqref{eq:uso}, one obtains
\begin{multline}\label{eq:sigma_p}
    \dot{\sigma}_{i}=-\frac{\gamma_{2}}{\hat{T}_{c_2}}\left(\alpha_{2}\abs{\sigma_{i}}^{p'}+\beta_{2}\abs{\sigma_{i}}^{q'}\right)^{k'}\sign{\sigma_{i}}\left(1+\frac{\abs{e_{v,i}}}{\abs{\barpow{e_{v,i}}^2+\frac{\gamma_1^2}{\hat{T}_{c_1}^2}\left(\alpha_1\barpow{e_{x,i}}^1+\beta_1\barpow{e_{x,i}}^3\right)}^{1/2}} \right)\\-\frac{\gamma_{1}}{2\hat{T}_{c_{1}}}\left(\alpha_1+3\beta_{1}e_{x,i}^{2}\right)\left(\sign{\sigma_{i}}+ \frac{\abs{e_{v,i}}\sign{\sigma_{i}}-e_{v,i}}{\abs{\barpow{e_{v,i}}^2+\frac{\gamma_1^2}{\hat{T}_{c_1}^2}\left(\alpha_1\barpow{e_{x,i}}^1+\beta_1\barpow{e_{x,i}}^3\right)}^{1/2}} \right)
    \\ -\left(\zeta_i \sign{\sigma_{i}}-\Delta_i+u_0\right) \left(1+\frac{\abs{e_{v,i}}}{\abs{\barpow{e_{v,i}}^2+\frac{\gamma_1^2}{\hat{T}_{c_1}^2}\left(\alpha_1\barpow{e_{x,i}}^1+\beta_1\barpow{e_{x,i}}^3\right)}^{1/2}} \right).
    \end{multline}
Let us consider the candidate Lyapunov function $V_{1}(\sigma_{i})=\abs{\sigma_{i}}$ with its time derivative as $\dot{V}_{1}=\sign{\sigma_{i}}\dot{\sigma_{i}}$. Since 
\begin{align}
    \frac{\abs{e_{v,i}}}{\abs{\barpow{e_{v,i}}^2+\frac{\gamma_1^2}{\hat{T}_{c_1}^2}\left(\alpha_1\barpow{e_{x,i}}^1+\beta_1\barpow{e_{x,i}}^3\right)}^{1/2}} \geq& 0,\\
     \abs{e_{v,i}}-e_{v,i}\sign{\sigma_{i}} \geq& 0,
    \end{align}
and  $$\zeta_i \geq u_0^{max}+\delta_i$$
one can easily rewrite the Lyapunov function derivative, using \eqref{eq:sigma_p}, in the following inequality
    \begin{align}\label{eq:deriv_lyap_v2}
    \dot{V}_{1}(\sigma) \leq -\frac{\gamma_{2}}{\hat{T}_{c_{2}}}\left(\alpha_{2}V_{1}(\sigma)^{p'}+\beta_{2} V_{1}(\sigma)^{q'} \right)^{k'}.
    \end{align}
From Theorem  \ref{thm:pred_time_stable}, one can deduce that $\sigma_i$ converges to zero in a fixed-time $\hat{T}_{c_{2}}$.

After time $\hat{T}_{c_2}$, one obtains
    \begin{align}
    0=e_{v,i}+\barpow{\barpow{e_{v,i}}^2+\frac{\gamma_1^2}{\hat{T}_{c_1}^2}\left(\alpha_1\barpow{e_{x,i}}^1+\beta_1\barpow{e_{x,i}}^3\right)}^{1/2},
    \end{align}
which in turn implies,

 \begin{align}
       \dot{e}_{x,i}=e_{v,i}=-\frac{\gamma_{1}}{\hat{T}_{c_{1}}}\left( \alpha_{1}\abs{e_{x,i}}+\beta_{1}\abs{e_{x,i}}^{3}\right)^{1/2}\sign{e_{x,i}}.
 \end{align}{}
 From Theorem \ref{th:cons_pred}, it is clear that $e_{x,i}$ converges to zero in a fixed-time before the settling time $\hat{T}_{c_{1}}$. Moreover, from \eqref{eq:sigmaso}, since $\sigma_i = 0$ and  $e_{x,i}=0$, then $e_{v,i}=0$. 
 Hence, we can conclude that system \eqref{eq:din_agent_i} with \eqref{eq:uso} as the control input, is fixed-time stable with predefined-time $\hat{T}_{c_1}+\hat{T}_{c_2}$. Moreover, due to Theorem \ref{thm:predefined_time_auto}, where the leader states are estimated in a fixed-time with the predefined settling time $T_{o}$. Hence, if $T_o'=T_o$, one can deduce that leader-follower consensus is achieved in fixed-time before the predefined-time $\hat{T}_c=T_{o}'+\hat{T}_{c_1}+\hat{T}_{c_2}$. At last, if $T'_{o}=0$ and $T_{o}<\hat{T}_{c_1}+\hat{T}_{c_2}$, the leader-follower consensus is  achieved before the predefined-time  $\hat{T}_c=\hat{T}_{c_1}+\hat{T}_{c_2}$.
\end{proof}

\section{Fixed-time leader-follower consensus with improved estimate for the \textit{UBST} using non-autonomous protocol }
\label{sec:NonAutoConsensus}
The autonomous leader-follower protocol presented in Section~\ref{Sec:AutConsensus} allows a fixed-time convergence. However, the estimate of the \textit{UBST} for the observer and controller are both too conservative. This is a common drawback on existing fixed-time consensus protocols, see e.g.,~\cite{Defoort2015,Ning2018} for the leader -follower problem for agents with first-order integrator dynamics and \cite{Ni2017,Ni2019} for agents with second-order integrator dynamics.  To address this issue, we present new protocols, based on the class of time-varying gains proposed in~\cite{Aldana2019}, to significantly reduce such conservatism.  Contrary to some existing protocols such as \cite{Wang2018,Morasso1997,Yucelen2019}, where the time-varying gains become singular when consensus is reached, in our approach the convergence is achieved with bounded time-varying gains in a user-defined time.

Before designing the proposed fixed-time leader-follower consensus protocol, let us define the following functions:
\begin{definition}\label{def:function_Phi}
Let us define the following
\begin{itemize}
    \item 
    $\Phi:\mathbb{R}_+\to \mathbb{R}_+ \cup\{+\infty\}\setminus\{0\}$  is a continuous function on $\mathbb{R}_+\setminus\{0\}$ that satisfies
\begin{itemize}
    \item $\int_0^{+\infty} \Phi(z)dz = 1$,
    \item $\Phi(\tau)<+\infty$, $\forall \tau\in\mathbb{R}_+\setminus\{0\}$, 
    \item is either non-increasing or locally Lipschitz on $\mathbb{R}_+\setminus\{0\}$.
\end{itemize}
\item $\psi:\mathbb{R}_+\to \mathbb{R}_+$ satisfies
$$\psi(\tau;T_c)=T_c\int_{0}^\tau\Phi(\xi)d\xi$$ with $T_c$ a positive constant,
\item $\eta$ is such that $$\eta(T)=\lim_{\tau\to T}\frac{1}{T_c}\psi(\tau;T_c)\leq 1$$
with $T$ a positive parameter,
\item \label{itm:function_rho} $\rho:\mathbb{R}_+\to \mathbb{R}_+$ satisfies
$$\rho(\tau;T_c) = \frac{1}{T_c}\Phi(\tau)^{-1}.$$
\end{itemize}
\end{definition}
\Blue{
\begin{definition} (Definition in \cite{kuhnel2015differential}). A regular parametrized curve, with parameter $t$, is a $\mathcal{C}^1(\mathcal{I})$ inmersion $c: \mathcal{I} \mapsto \mathbb{R}$, defined on a real interval $I\subseteq \mathbb{R}$. This means that $\frac{dc}{dt}\neq0$ holds everywhere
\end{definition}
\begin{definition}\label{def:parametrized} (Pg. 8 in \cite{kuhnel2015differential}). A regular curve is an equivalence class of regular parametrized curves, where the equivalence relation is given by regular (orientarion preserving) parameter transformation $\psi$, where $\psi:\mathcal{I} \rightarrow I'$ is $\mathcal{C}^{1}(\mathcal{I})$, bijective, and $\frac{d\psi}{dt}>0$. Therefore, if $c: \mathcal{I} \rightarrow \mathbb{R}$ is a regular parametrized curve and $\psi:\mathcal{I}\rightarrow \mathcal{I}'$ is a regular parameter transformation, then $c$ and $c \circ \psi : \mathcal{I}' \rightarrow \mathbb{R}$ are considered to be equivalent.
\end{definition}}
\begin{lemma}\label{lemm:time_trans}\cite{gomez2020design}
Let $t_0$ be the initial time. The function $t = \psi(\tau)+t_{0} $, defines a parameter transformation with $\tau = \psi^{-1}(t-t_{0})$ as its inverse mapping.  
\end{lemma}
\Blue{\begin{proof}
It follows from Definition \ref{def:parametrized}
\end{proof}}

To derive the fixed-time non-autonomous scheme, we define the following time-varying gain for each predefined settling time $T_{c_{i}}$ as follows using the previously defined functions:
\begin{equation}
\hat{\kappa}_{i}(t;t_0,T_{c_i},T)=\left\lbrace
\begin{array}{ccc}
\rho_{i}(\psi_{i}^{-1}(t-t_0;T_{c_i});T_{c_i}) & \text{if} & t\in [t_0,t_0+\eta_{i}(T) T_{c_i})\\
1 & & \text{otherwise}.
\end{array}
\right.
\end{equation}

\begin{remark}
Notice that if $T<+\infty$, then $\hat{\kappa}_{i}(t;t_0,T_{c_i},T)$ is bounded. Such bound can be user-defined by tuning $T$.
\end{remark}

Now, we are ready to present our main result. 
\begin{theorem}\label{thm:non_aut_protocol}
Let us consider the same observer parameters as in Theorem \ref{thm:predefined_time_auto} (i.e. $\alpha, \beta, p, q, k > 0, \zeta_x, \zeta_v,  \kappa_{i,x},  \kappa_{i,v}$, $T_{c_1}$, $T_{c_2}$).
Let $\varrho_1(t)=\hat{\kappa}_{1}(t;t_0,T_{c_1},T_{\alpha})$ and  $\varrho_2(t)=\hat{\kappa}_{2}(t;t'_{0},T_{c_2},T_{\beta})$ be time-varying gains. $T_{\alpha}$ and $T_{\beta}$ are positive parameters and $t'_{0}=t_{0}+\eta_{1}(T_{\alpha})T_{c_1}$.

Using the following distributed observer
\begin{align}\label{eq:observer_non_aut}
 \dot{\hat{x}}_{i}&=\hat{v}_i - \varrho_2(t)\kappa_{i,x}\left[(\alpha|e_{1,i}|^{p} + \beta|e_{1,i}|^{q})^k+\hat{\zeta}_{x}(t)\right] \sign{e_{1,i}}\\
 \dot{\hat{v}}_i &= -\varrho_1(t)\kappa_{i,v}\left[(\alpha|e_{2,i}|^{p} + \beta|e_{2,i}|^{q})^k+\hat{\zeta}_{v}(t)\right] \sign{e_{2,i}},
\end{align}
with $\hat{\zeta}_{x}(t)=\varrho_{2}^{-1}(t)\zeta_{x}$ and $\hat{\zeta}_{v}(t)=\varrho_{1}^{-1}(t)\zeta_{v}$, the observer error dynamics is fixed-time stable with the \textit{UBST} given by $T_{o}=t_{0}+\eta_{1}(T_{\alpha})T_{c_1}+\eta_{2}(T_{\beta})T_{c_2}$, with $T_{\alpha},T_{\beta}>0$. 

Let us consider the same control parameters as in Theorem \ref{thm:predefined_time_cont_auto} (i.e.  $\alpha_1,\alpha_2,\beta_1,\beta_2,\hat{T}_{c_1},\hat{T}_{c_2}>0$, $p', q', k' > 0,\zeta_i,\gamma_{1}, \gamma_{2} $) and set $T_{c_3}=\hat{T}_{c_1}+\hat{T}_{c_2}$. Let $\varrho_3(t)=\hat{\kappa}_{3}(t;T_{o},T_{c_3},T_{\gamma})$ be a time-varying gain with $T_{\gamma}$ a positive parameter.

Using the following controller
\begin{align}\label{eq:control_pred}
    u_{i}= \left\{ \begin{array}{ccc}
         \varrho_3(t)^2\upsilon(e_{x,i},\varrho_3(t)^{-1}e_{v,i})+\dot{\varrho}_3(t)\varrho_3(t)^{-1}e_{v,i}&  \text{   if } & t \in [T'_{o},T'_{o}+\eta_{3} (T_{\gamma})T_{c_3})\\
         \upsilon(e_{x,i},e_{v,i}) & \text{   if } & t \in [T'_{o}+\eta_{3}(T_{\gamma})T_{c_3}, +\infty) 
    \end{array}\right.
\end{align}
where $\upsilon(e_{x,i},\varrho_{3}^{-1}(t)e_{v,i})$ is given by~\eqref{eq:uso}, the leader-follower consensus is achieved in fixed time with the \textit{UBST} given by $\hat{T}=T'_{o}+\eta(T_{\gamma})T_{c_3}$.

\end{theorem}

The proof of Theorem \ref{thm:non_aut_protocol} will be divided into two parts. The first part focuses on the observer stability whereas the second one focuses on the controller stability.

\begin{proof}
 
First, let us study the observer error dynamics using the non-autonomous observer \eqref{eq:observer_non_aut}. 

Let us consider the observer errors as in \eqref{eq:error_obs_leader}. Using \eqref{eq:observer_non_aut}, the observation error dynamics is expressed as follows
\begin{align}\label{eq:error_observer_tau}
     \dot{\tilde{x}}_{i}&=\tilde{v}_i - \varrho_{2} \kappa_{i,x}\left[(\alpha|e_{1,i}|^{p} + \beta|e_{1,i}|^{q})^k+\hat{\zeta}_{x}\right] \sign{e_{1,i}}\\
    \dot{\tilde{v}}_i &= -\varrho_{1}\kappa_{i,v}\left[(\alpha|e_{2,i}|^{p} + \beta|e_{2,i}|^{q})^k+\hat{\zeta}_{v}\right] \sign{e_{2,i}}-u_{0}.
\end{align}
Now, considering the observer error dynamics of the velocity $\tilde{v}_{i}$ in the new $\tau$-time variable  as follows
\begin{align}\label{eq:dyn_v_obs_tau}
    \frac{d \tilde{v}_{i}}{d \tau}&= \frac{d \tilde{v}_{i}}{d t}\frac{d t}{d \tau},
\end{align}
and according to the parameter transformation given in Lemma \ref{lemm:time_trans},
\begin{align}\label{eq:dt_dtau}
    \frac{d t}{d \tau}=\left.\frac{d}{d \tau}\left(\psi_{i}(\tau)-t_{0}\right)\right|_{\tau=\psi_{i}^{-1}(t-t_{0};T_{c_{i}})}&=\rho_{i}(\psi_{i}^{-1}(t-t_{0});T_{c_{i}})^{-1}\\&=\hat{\kappa}_{i}(t;t_{0},T_{c_{i}},T)^{-1}.
\end{align}
Thus, the observation error dynamics of the velocity given by \eqref{eq:dyn_v_obs_tau} is rewritten, using \eqref{eq:dt_dtau}, as follows 
\begin{align} \label{eq:dyn_v_obs_tau_2}
    \frac{d \tilde{v}_{i}}{d \tau}&= -\hat{\kappa}_{1}(t;t_0,T_{c_1},T_{\alpha})^{-1}\varrho_{1}\kappa_{i,v}\left[(\alpha|e_{2,i}|^{p} + \beta|e_{2,i}|^{q})^k+\hat{\zeta}_{v}\right]\sign{e_{2,i}} - \varpi(\tau),
\end{align}
where $\varpi(\tau)= \hat{\kappa}_{1}(t;t_0,T_{c_1},T_{\alpha})^{-1}u_{0}$ with $\abs{u_{0}}<u_{0}^{max}$ is the disturbance term and $\hat{\kappa}_{1}(t;t_0,T_{c_1},T_{\alpha})^{-1}=\varrho_{1}^{-1}(t)$. Then, the last expression is written as  

\begin{align}\label{eq:dyn_obs_tau}
        \frac{d \tilde{v}_{i}}{d \tau}&=- \kappa_{i,v}\left[(\alpha|e_{2,i}|^{p} + \beta|e_{2,i}|^{q})^k+\hat{\zeta}_{v}\right]\sign{e_{2,i}}-\varpi(\tau).
\end{align}
Note that by the definition of the function $\Phi_{1}(\tau)$, the time-varying gain $\hat{\kappa}_{1}(t;t_0,T_{c_1},T_{\alpha})^{-1}$ for $t \in [t_{0},t_{0}+\eta_{1}(T_{\alpha})T_{c_1})$, can be written as  $\rho_{1}(\tau;T_{c_1})^{-1} \; \forall \tau$, and function $\rho_{1}(\tau;T_{c_1})^{-1}$ is non-increasing. Besides, $\rho_{1}(\tau;T_{c_1})^{-1}$ is bounded and $\rho_{1}(\tau;T_{c_1})^{-1}\rightarrow 0$ as $\tau \rightarrow +\infty$. Then, the disturbance $\varpi(\tau)=\rho_{1}(\tau;T_{c_1})^{-1}u_{o}$ is vanishing. Furthermore, notice that $\hat{\zeta}_{v}(t)=\varrho_{1}^{-1}(t)\zeta_{v}$ and $\abs{\varpi(\tau)}<\hat{\zeta}_{v},\; \forall \tau$ since $u_{0}^{max}\leq\kappa_{v}\zeta_{v}$.
Then, similarly to \eqref{eq:compact_form}, the compact form of \eqref{eq:dyn_obs_tau} is 
\begin{align}
\label{eq:compact_form_tau_v}
    \frac{d\tilde{v}}{d\tau} &= -\Phi_v\left(\graphm{}\tilde{v}\right)-  \boldsymbol{1} \varpi(\tau)
\end{align}
with $\tilde{v}=[\tilde{v}_1 \ \cdots \ \tilde{v}_N]^T\in\mathbb{R}^N$. 

Furthermore, from Theorem \ref{thm:predefined_time_auto}, the observation error dynamics of velocity \eqref{eq:compact_form_tau_v} is fixed-time stable in the time-variable $\tau$ and,  converges to the origin with a settling time $T'_{c_{1}}<+\infty$. Then, the observation error in velocity reaches the origin at $T(\tilde{v}_{0})=\lim_{\tau\rightarrow T'_{c_{1}}}(\psi_1(\tau)+t_{0})\leq t_{0}+\eta_{1}(T_{\alpha})T_{c_1};\forall \tilde{v}_{0} \in  \mathbb{R}^{N}$ as the initial conditions. In a similar way, the observation error dynamics of the position $\tilde{x}_{i}$ in the time-variable $\tau$ is written as follows  
\begin{align}\label{eq:dyn_x_obs_tau}
\frac{d \tilde{x}_{i}}{d \tau}&= \upsilon_{i}(\tau)- \kappa_{i,x}\left[(\alpha|e_{2,i}|^{p} + \beta|e_{2,i}|^{q})^k+\hat{\zeta}_{x}\right]\sign{e_{2,i}}
\end{align}
where $\upsilon_{i}(\tau)=\hat{\kappa}_{2}(t;t_{0}+\eta_{1}(T_{\alpha})T_{c_1},T_{c_2},T_{\beta})^{-1}\tilde{v}_{i}$. The compact form of \eqref{eq:dyn_x_obs_tau} is 
\begin{align}
\label{eq:compact_form_tau_x}
\frac{d\tilde{x}}{d\tau} &= \upsilon -\Phi_x\left(\graphm{}\tilde{x}\right)
\end{align}
with $\upsilon=[\upsilon_{1}\cdots\upsilon_{N}]\in \mathbb{R}^{N}$. Then, due to $\tilde{v}=0$ for $t\geq t_{0}+\eta_{1}(T_{\alpha})T_{c_1}$, the observation error dynamics of position \eqref{eq:compact_form_tau_x} is fixed-time stable in the time-variable $\tau$, and converges to the origin with a settling time $T'_{c_{2}}<+\infty$. Hence, the observation error in position reaches the origin at $T(\tilde{x}_{0})=\lim_{\tau\rightarrow T'_{c_{2}}}(\psi_2(\tau)+t'_{0})\leq t'_{0}+\eta_{2}(T_{\beta})T_{c_2};\forall \tilde{x}_{0} \in  \mathbb{R}^{N}$, with $t'_{0}=t_{0}+\eta_{1}(T_{\alpha})T_{c_1}$. Therefore, the observer error dynamics is fixed-time stable and converges to the origin before the predefined-time $T_{o}=t_{0}+\eta_{1}(T_{\alpha})T_{c_1}+\eta_{2}(T_{\beta})T_{c_2}$. Notice that $T'_{c_{1}}$ and $T'_{c_{2}}$ are the settling time for the system in the time-variable $\tau$. 

Then, let us study the tracking error dynamics using the non-autonomous controller \eqref{eq:control_pred}. Consider the following coordinate change
\begin{align}\label{eq:change_coordinate}
e_{x,i}=&\tilde{e}_{x,i}\\
e_{v,i}=&\varrho_{3}(t)\tilde{e}_{v,i}
\end{align}
where $e_{x,i}$ and $e_{v,i}$ are the tracking errors for each agent defined in \eqref{eq:errors} or, its equivalent in \eqref{eq:errors2} after the convergence of the observation error. Then, the dynamics of the variable $\tilde{e}_{i}=[\tilde{e}_{x,i},\tilde{e}_{v,i}]^T$ is the following
\begin{align}\label{eq:}
    \dot{\tilde{e}}_{x,i}&=\varrho_{3}\tilde{e}_{v,i}\\
    \dot{\tilde{e}}_{v,i}&=\varrho_{3}^{-1}u_{i}-\varrho_{3}^{-1}\dot{\varrho}_{3}\tilde{e}_{v,i}+\varrho_{3}^{-1}(\Delta_{i}-u_{0}).
\end{align}
      
Now, let us consider the parameter transformation given by Lemma \ref{lemm:time_trans} to get the dynamics in $\tau$-variable. Then, the dynamics of \eqref{eq:change_coordinate} expressed in the time-variable $\tau$ is

\begin{align}\label{eq:sys_tau_track}
     \frac{d \tilde{e}_{x,i}}{d \tau} &= \tilde{e}_{v,i}\\
 \frac{d \tilde{e}_{v,i}}{d \tau}&= \varrho_{3}^{-2}u_{i}-\varrho_{3}^{-2}\dot{\varrho}_{3}\tilde{e}_{v,i}+\varrho_{3}^{-2}(\Delta_{i}-u_{0}).
\end{align}

Using the controller \eqref{eq:control_pred} for $t\in [T_{o}',T_{o}'+\eta_{3}(T_{\gamma})T_{c_3})$, system \eqref{eq:sys_tau_track} is written as

 \begin{align}\label{eq:din_tau}
     \frac{d \tilde{e}_{x,i}}{d \tau} &= \tilde{e}_{v,i}\\
 \frac{d \tilde{e}_{v,i}}{d \tau}&= \upsilon (\tilde{e}_{x,i},\tilde{e}_{v,i}) + \pi_{i}(\tau),
 \end{align}
with $\pi_{i} (\tau)=\left.\varrho_{3}^{-2}(\Delta_{i}(t)-u_{0})\right|_{t=\psi_{3}(\tau)+T_{o}}$. Notice that $\Delta_{i}(t)$ satisfies $\abs{\Delta_{i}(t)}\leq \delta_{i}$ and $u_{0}$ is unknown but bounded by $\abs{u_{o}}\leq u_{0}^{max}$. By Definition \ref{def:function_Phi}, the function $\varrho_{3}(t)^{-2}$ is non increasing. Then, $\pi_{i}(\tau)$ is bounded and $\pi_{i}(\tau)\rightarrow 0$ as $\tau \rightarrow +\infty$. Thus, from Theorem \ref{thm:predefined_time_cont_auto}, system \eqref{eq:din_tau} is fixed-time stable in the time-variable $\tau$ with $T'_{c_{3}}<+\infty$ as its settling time. Hence, the tracking errors ($\tilde{e}_{x,i}$ and $\tilde{e}_{v,i}$), with $e_{0}=[\tilde{e}_{x,i}(T_{o}),\tilde{e}_{v,i}(T_{o})]$ as initial conditions, reach the origin at $T(e_{0})=\lim_{\tau\rightarrow T'_{c3}}(\psi_3(\tau)+T_{o})\leq T_{o}+\eta_3 (T_{\gamma})T_{c_3}; \forall e_{0}\in \mathbb{R}^{2}$.

 Thus, the tracking error dynamics is fixed-time stable with $\eta_{3} (T_{\gamma})T_{c_3}$ as the predefined \textit{UBST}. Hence, if $T'_{o}=T_{o}$ and from the fact that observer \eqref{eq:observer_non_aut} estimates the leader state in predefined-time and controller \eqref{eq:control_pred} drives the agents towards the leader state trajectory, one can conclude that the leader-follower consensus is achieved in fixed-time before the predefined-time $\hat{T}=T_{o}+\eta(T_\gamma)T_{c_3}$. At last, if $T'_{o}=t_{o}$ and $T_{o}<\eta(T_\gamma)T_{c_3}$ the leader-follower consensus is achieved before the predefined time $\hat{T}=t_{0}+\eta(T_\gamma)T_{c_3}$.
\end{proof}
\begin{remark}
It is worth noting that the non-autonomous protocol is derived from the autonomous one. As discussed in the next section, this scheme based on bounded time-varying gains, has been introduced to improve the convergence time estimation, i.e.,  the slack between the \textit{UBST} and the convergence time is reduced.
\end{remark}

\section{Simulation results}\label{sec:results}
In this section, we illustrate our main results with the autonomous and non-autonomous protocols for the leader-follower consensus problem. In order to compare the autonomous and non-autonomous schemes proposed in this work, we will also compare the slack between the \textit{UBST} and the real convergence time of the system for each control scheme.
 \begin{figure}[h]
 	\centering
 	\includegraphics[width=0.4\textwidth]{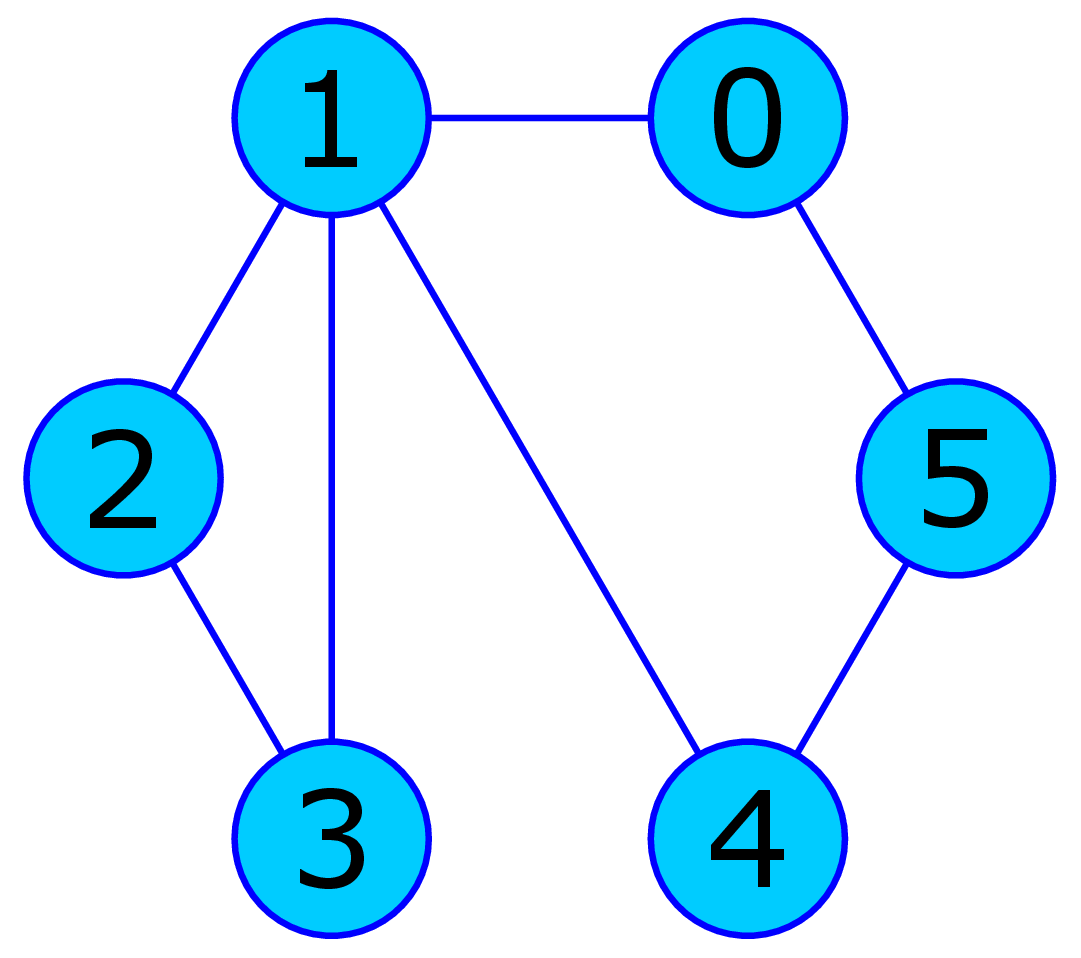}
 	\caption{Communication topology with 5 followers.}
 	\label{fig:graph_2}
 \end{figure}
For all schemes, we consider the same scenario for comparison purposes. We consider a multi-agent system composed of $N = 5$ agents where the dynamics of each agent is given by Eq. \eqref{eq:din_agent_i} with an external perturbation  $\Delta_i(t)= \sin (40t+0.1i)$, with $i=1\hdots N$. The communication topology, given in Figure \ref{fig:graph_2} is undirected and contains a spanning tree with the leader agent as the root. For the leader, its control input is given by $u_{0}=4\cos{(2t)}$ with the initial conditions $[x_{0},v_{0}]=[-1,0]$. From Figure \ref{fig:graph_2}, one gets  $\lambda_{\min}(\graphm{})=0.2907$. The initial conditions of the agents are as follows 
\begin{align}
  x(0)&=\left[-10, -5, 0, 5, 10\right]\\
  v(0)&=\left[0, 0, 0, 0, 0\right]
\end{align}
and the initial conditions for the observer are randomly set as 
\begin{align}
    \hat{x}(0)&=\left[-5.81, -7.82, 4.57, 9.22, 5.94\right]\\
    \hat{v}(0)&=\left[5.57, -6.42, 4.91, 8.39, -7.87\right].
\end{align}

\subsection{Fixed-time leader-follower consensus using autonomous protocols }\label{sec:result_fixed}

In this subsection, we present the results of the autonomous scheme presented in Section \ref{Sec:AutConsensus}. According to Theorem \ref{thm:predefined_time_auto}, the distributed fixed-time observer \eqref{eq:observer_dynamics} with  $p=1.5$, $q=3.0$, $k=0.5$, $\alpha=1$, $\beta=2$, $T_{c_{1}}=0.9s$, $T_{c_{2}}=0.1$, $\kappa_{x}=3.53$, $\kappa_{v}=31.82$ and  $\zeta_{v}= 0.0678$ guarantees that the observer error converges to zero before the predefined-time $T_{c_{1}}+T_{c_{2}}=1s$. Figure \ref{fig:obs_auto_1} shows the leader state estimation for each agent, while the left column of Figure \ref{fig:observer_errors} shows the observer errors for each agent. One can see in more details in the left column of Figure \ref{fig:observer_errors_zoom} the settling time of the observation errors and the \textit{UBST} as the dotted line, where the settling time for the velocity error ($\tilde{v}$) is $T_{1}\approx 0.013s$, and for the position error ($\tilde{x}$) is  $T_{2}\approx0.143s$. It is possible to see that the settling time for the observation error occurs before the predefined time, i.e., $T_{1}<T_{c_1}$ for the velocity error and $T_{1}+T_{2}<T_{o}$ for the observation error. In the simulation the controller is activated at the same time that the observer, i.e., $T'_{o}=0$ and $T_{o}<\hat{T}_{c_{1}}+\hat{T}_{c_{2}}$. Then, the controller \eqref{eq:uso} is applied in order to follow the trajectory of the leader with $p'=1.5$, $q'=3.0$, $k'=0.5$, $\alpha_{1}=\alpha_{2}=1/\beta_{1}=1/\beta_{2}=1/4$, $\hat{T}_{c_{1}}=1s$ and $\hat{T}_{c_{2}}=1s$. The left column of Figure \ref{fig:tracking_errors} shows the tracking error, where the tracking errors $e_{x}$ and $e_{v}$ converge to zero before $\hat{T}_{c}=\hat{T}_{c_{1}}+\hat{T}_{c_{2}}=2s$ with $T'_{o}=0$. One can see in more details in Figure \ref{fig:tracking_errors_zom} that the convergence time of the tracking error occurs at $T_{3}\approx 1.228s<\hat{T}_{c}$ where the \textit{UBST} is plotted in dotted line. The states of the agents are shown in Figure \ref{fig:sys_states_auto}, where it can be seen that the leader-follower consensus is successfully achieved.

   \begin{figure}[h!]
 	\centering
 	\includegraphics[width=\textwidth]{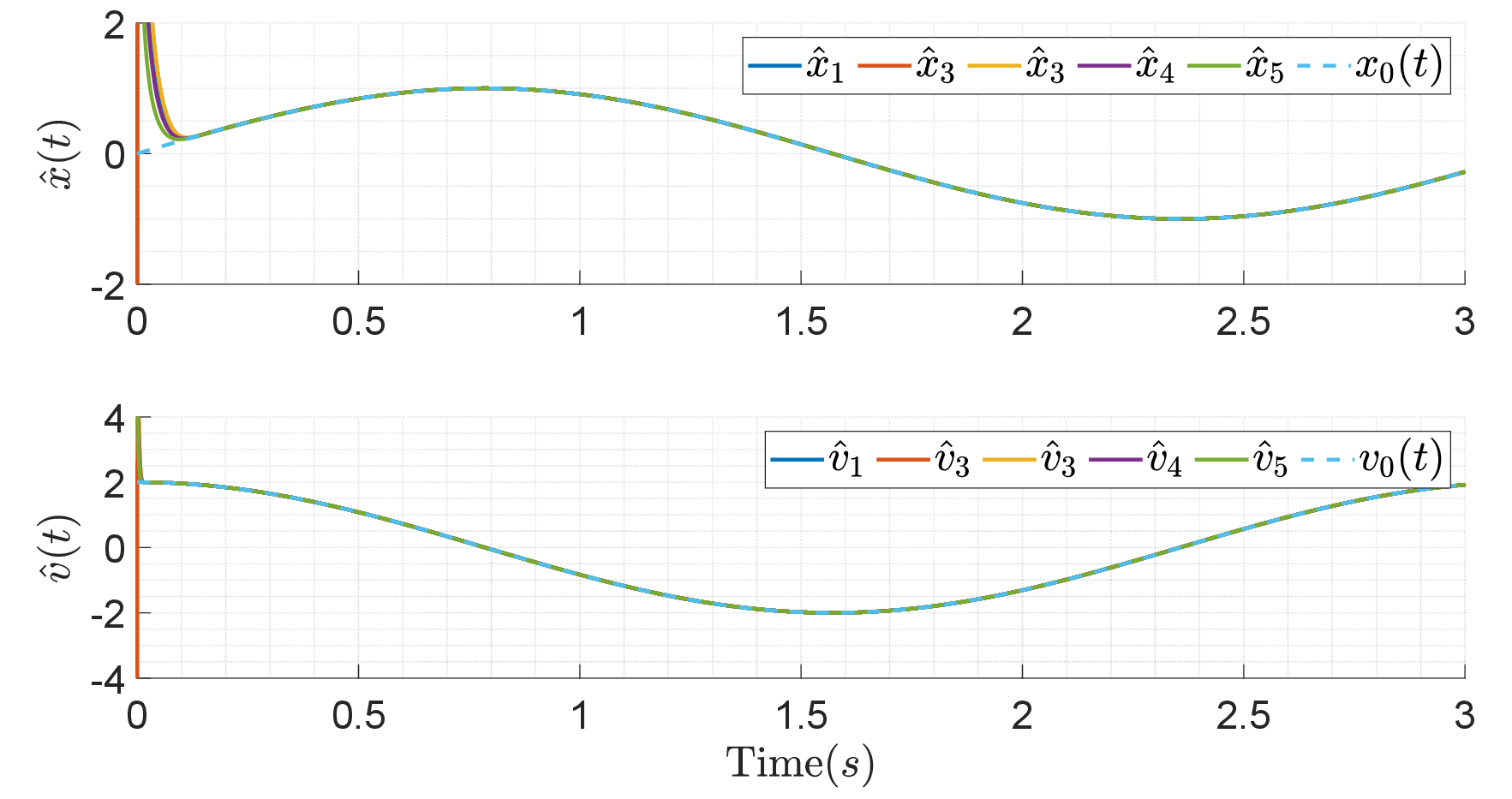}
 	\caption{\textbf{Autonomous protocol.} Leader states estimation for each agent.}
 	\label{fig:obs_auto_1}
 \end{figure}

 \begin{figure}[h!]
 	\centering
\includegraphics[width=\textwidth]{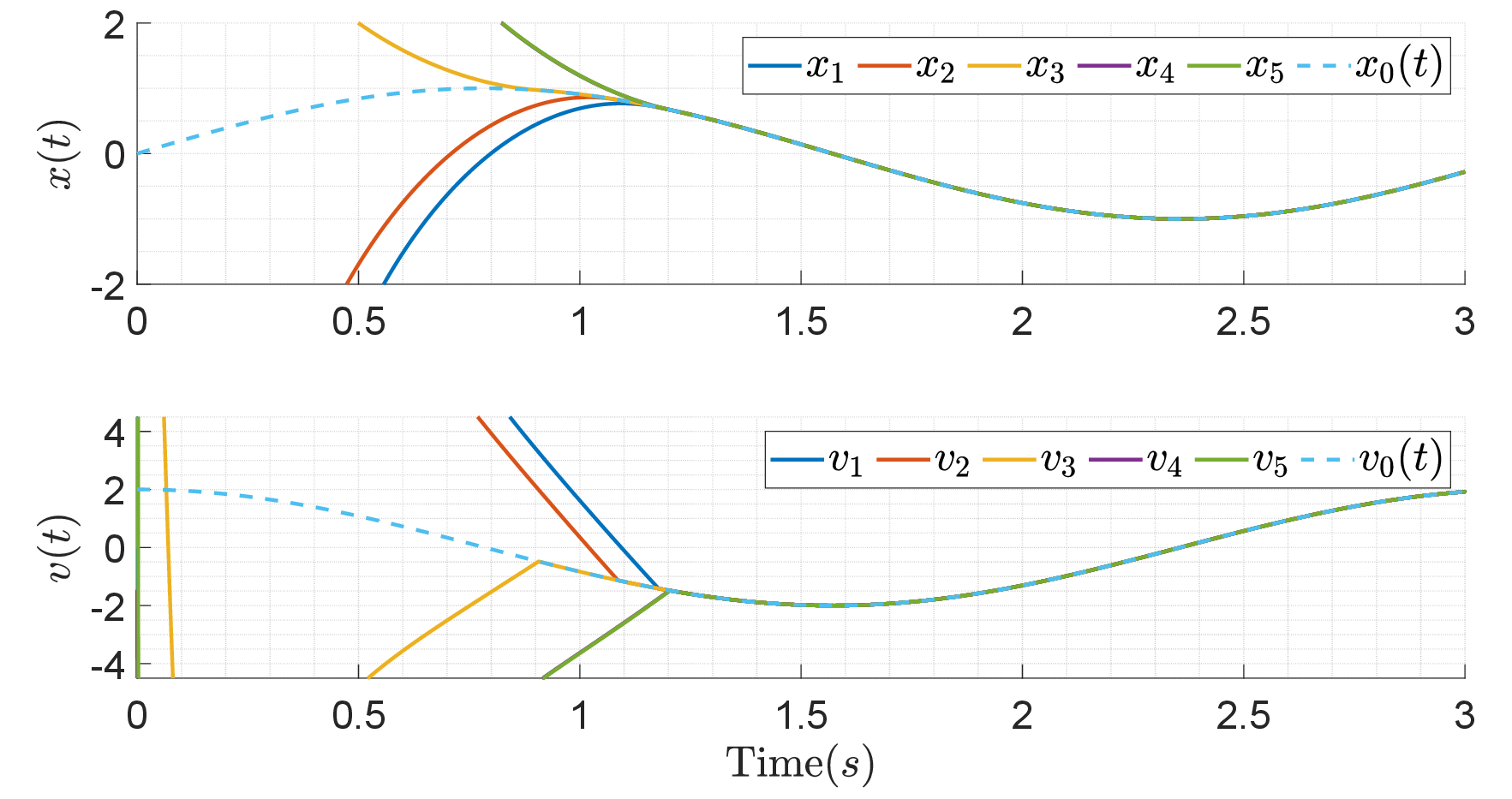}

 	\caption{ \textbf{Autonomous protocol.} Trajectories of each agent.}
 	\label{fig:sys_states_auto}
 \end{figure}

\subsection{Fixed-time leader-follower consensus using non-autonomous protocols}
The results of the fixed-time observer/controller scheme presented in Section \ref{sec:NonAutoConsensus}, using time-varying gains, are shown in this subsection. Consider the same example as in the previous subsection with the same external perturbation. 

Now, consider the observer/controller scheme proposed in Theorem \ref{thm:non_aut_protocol}, where function $\Phi_i$ is defined as $$\Phi_{i}(\tau)=\hat{\alpha}_{i}\eta_{i}^{-1} e^{-\hat{\alpha}_{i} \tau}$$ for $i=1,2,3$ with $\eta_{i}(T)=1-e^{-\hat{\alpha}_{i}T}$  and $\hat{\alpha}_{1}=220$, $\hat{\alpha}_{2}=90$ and $\hat{\alpha}_{3}=1.8$. This function is used to compute the time-varying gains $\varrho_{i}$. Then, the gain $\hat{\kappa}_{i} (t;t_{o},T_{c_i},T)$ used for the observer \eqref{eq:observer_non_aut} and the controller \eqref{eq:control_pred}, is defined as

\begin{align}
\hat{\kappa}_{i}(t;t_0,T_{c_i},T)=\left\lbrace
\begin{array}{ccc}
\frac{\eta_{i}}{\hat{\alpha}_{i} (T_{c_{i}}-\eta_{i} (t-t_{0}))} & \text{if} & t\in [t_0,t_0+\eta_{i}(T)T_{c_i})\\
1 & & \text{otherwise}.
\end{array}
\right.
\end{align}    
The user-defined parameters are set as follows $T_{c_1}=0.1s$, $T_{c_2}=0.9s$ and $T_{c_{3}}=\hat{T}_{c_{1}}+\hat{T}_{c_{2}}$  with $\hat{T}_{c_{1}}=1s$ and $\hat{T}_{c_{2}}=1s$, the parameters for $\eta_{i}$ are set as $T_{\alpha}=0.016$, $T_{\beta}=0.055$ and $T_{\gamma}=1.5$ for $i=1,2,3$ respectively, $t_{o}=0s$ for the observer and $T'_{o}=t_{0}$ for the controller. In order to compare the control scheme proposed in Theorems \ref{thm:predefined_time_auto}-\ref{thm:predefined_time_cont_auto} with Theorem \ref{thm:non_aut_protocol}, the parameters $\alpha_1,\alpha_2,\beta_1,\beta_2$, $p', q', k' $ for the controller and the parameters $\alpha, \beta, p, q, k , \zeta_x, \zeta_v,  \kappa_{x},  \kappa_{v}$ for the observer were taken from the result presented in Section \ref{sec:result_fixed}. Figure \ref{fig:pre_obs_state} shows the leader state estimation for each agent, while the right column of Figure \ref{fig:observer_errors} shows the observer errors for each agent. One can see in more details in the right column of Figure \ref{fig:observer_errors_zoom} the settling time of the observation errors ($\tilde{x}$ and $\tilde{v}$) and the \textit{UBST} as the dotted line, where the settling time for the velocity error $\tilde{v}$ is $T_{1}\approx0.09478s$, and for the position error $\tilde{x}$ is  $T_{2}\approx0.9891$. It is possible to see that the settling time for the observation error occurs before the predefined time $T_{o}=\eta_{1}(T_{\alpha})T_{c_{1}}+\eta_{2}(T_{\beta})T_{c_{2}}$. Besides, the controller \eqref{eq:control_pred} is applied in order to follow the trajectory of the leader. The right column of Figure \ref{fig:tracking_errors} shows the fixed-time convergence of the tracking error ($e_{x}$ and $e_v$).  One can see in more details in the right column of Figure \ref{fig:tracking_errors_zom} that the convergence time of the tracking error occurs at $T_{3}\approx1.952s$.  The states of the agents are shown in Figure \ref{fig:sys_trajectories_non_auto}, where it can be seen that the leader-follower consensus is successfully achieved.

Unlike the control scheme proposed in Theorem \ref{thm:predefined_time_auto} (for the observer) and Theorem \ref{thm:predefined_time_cont_auto} (for the controller), the slack between the predefined \textit{UBST} given by the user and the real convergence time is less conservative. Figure \ref{fig:observer_errors_zoom} shows the convergence of both protocols for the observer and Figure \ref{fig:tracking_errors_zom} shows the convergence of both protocols for the controller, where one can see the \textit{UBST} as the dotted line and, the difference between the slack of the autonomous and non-autonomous protocols. Moreover, the slack of the non-autonomous protocol can be adjusted by the parameters of the time-varying gain $\hat{\kappa}_{i}(t;t_{0},T_{c_{i}},T)$. However, in this case, this gain increases. Thus, one needs to establish a trade-off between the size of the upper bound for $\hat{\kappa}_{i}(t;t_{0},T_{c_i},T)$ and how small the slack is.  
 \begin{figure}[h!]
 	\centering
  	 \includegraphics[width=\textwidth]{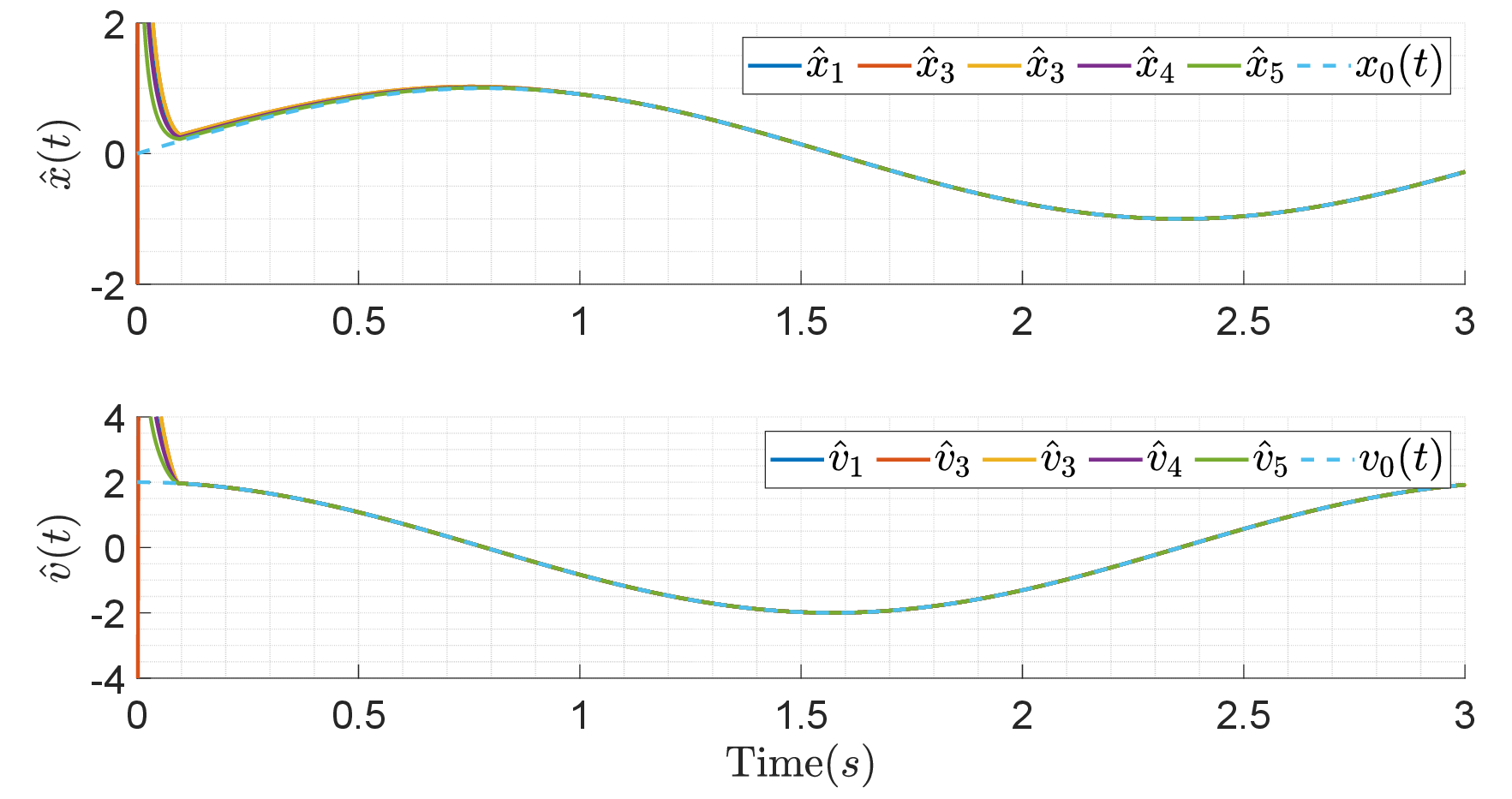}
 	\caption{ \textbf{Non-autonomous protocol.} Leader states estimation for each agent.}
 	\label{fig:pre_obs_state}
 \end{figure}

   \begin{figure}[h!]
 	\centering
\includegraphics[width=\textwidth]{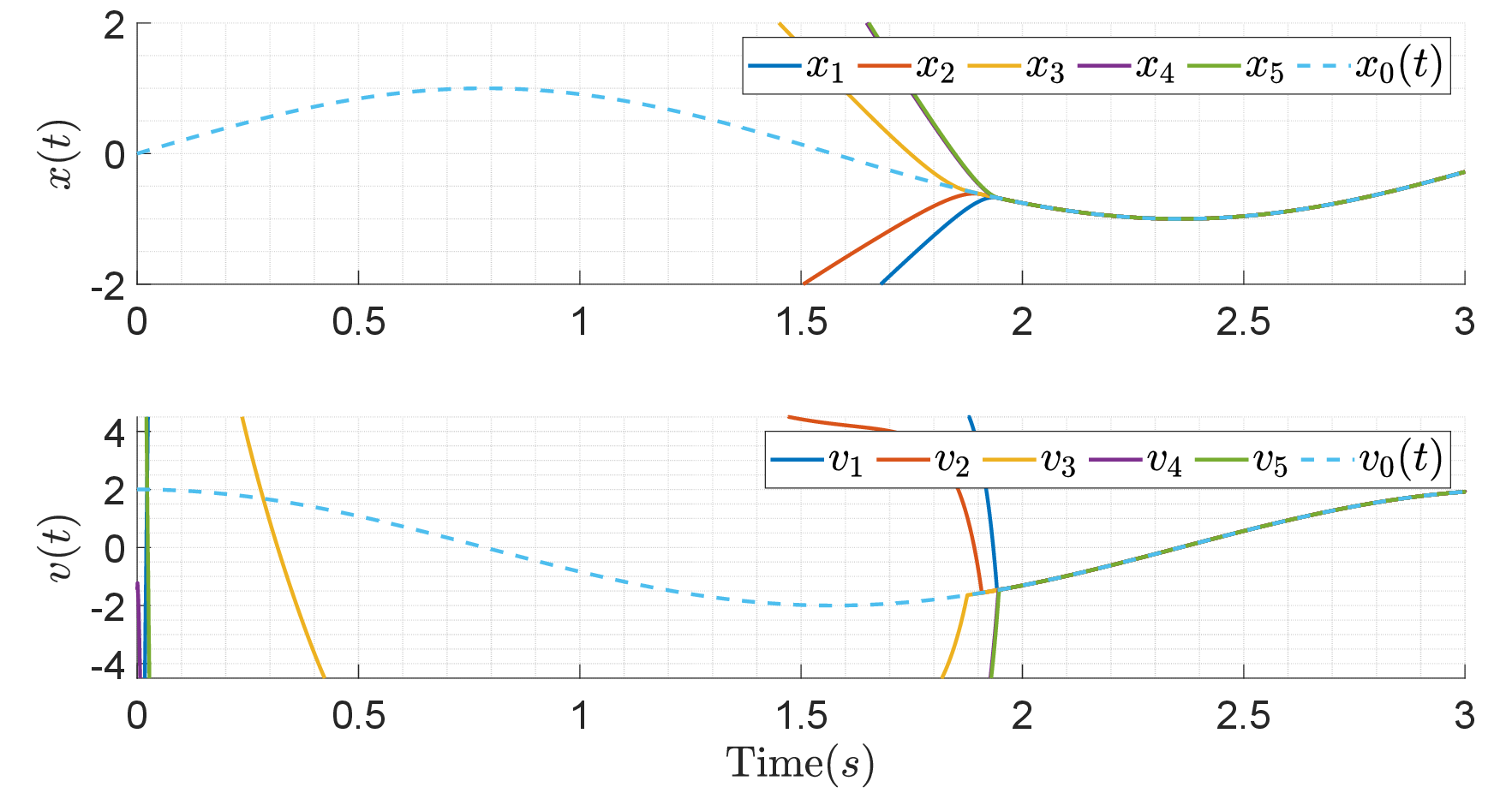}

 	\caption{ \textbf{Non-autonomous protocol.} Trajectories of each agent.}
 	\label{fig:sys_trajectories_non_auto}
 \end{figure}

In order to compare the two protocols presented in this work (autonomous and non-autonomous ones), we will define the slack function $s(x)$ as the error between the predefined \textit{UBST} and the convergence time of the variable $x$, i.e., $s(x)=T_{c}-T$ where $T_{c}$ is the predefined \textit{UBST} and $T$ is the actual convergence time. This index for the observation error and the tracking error variables is shown in Table \ref{tab:comparasion}. It can be seen that the slack for the non-autonomous protocol is lower than for the autonomous protocol.  
\begin{table}[H]
\centering
\begin{tabular}{||c |c |c| c||} 
 \hline
  & $s(\tilde{v})=T_{c_{1}}-T_{1}$ & $s(\tilde{x})=T_{o}-T_{2}$ & $s(\tilde{e})=\hat{T}-T_{3}$  \\ [0.5ex] 
 \hline\hline
 Autonomous protocol &0.08709  &0.8744 &0.743 \\ 
 Non-autonomous protocol & 0.0023 & 0.0016 & 0.1154 \\ [1ex]
 \hline
\end{tabular}
\caption{Slack of convergence time.}
\label{tab:comparasion}
 \end{table}
 These numerical examples show the effectiveness of the proposed consensus protocols.

\begin{figure}[h!]
 	\centering
 	\includegraphics[width=\textwidth]{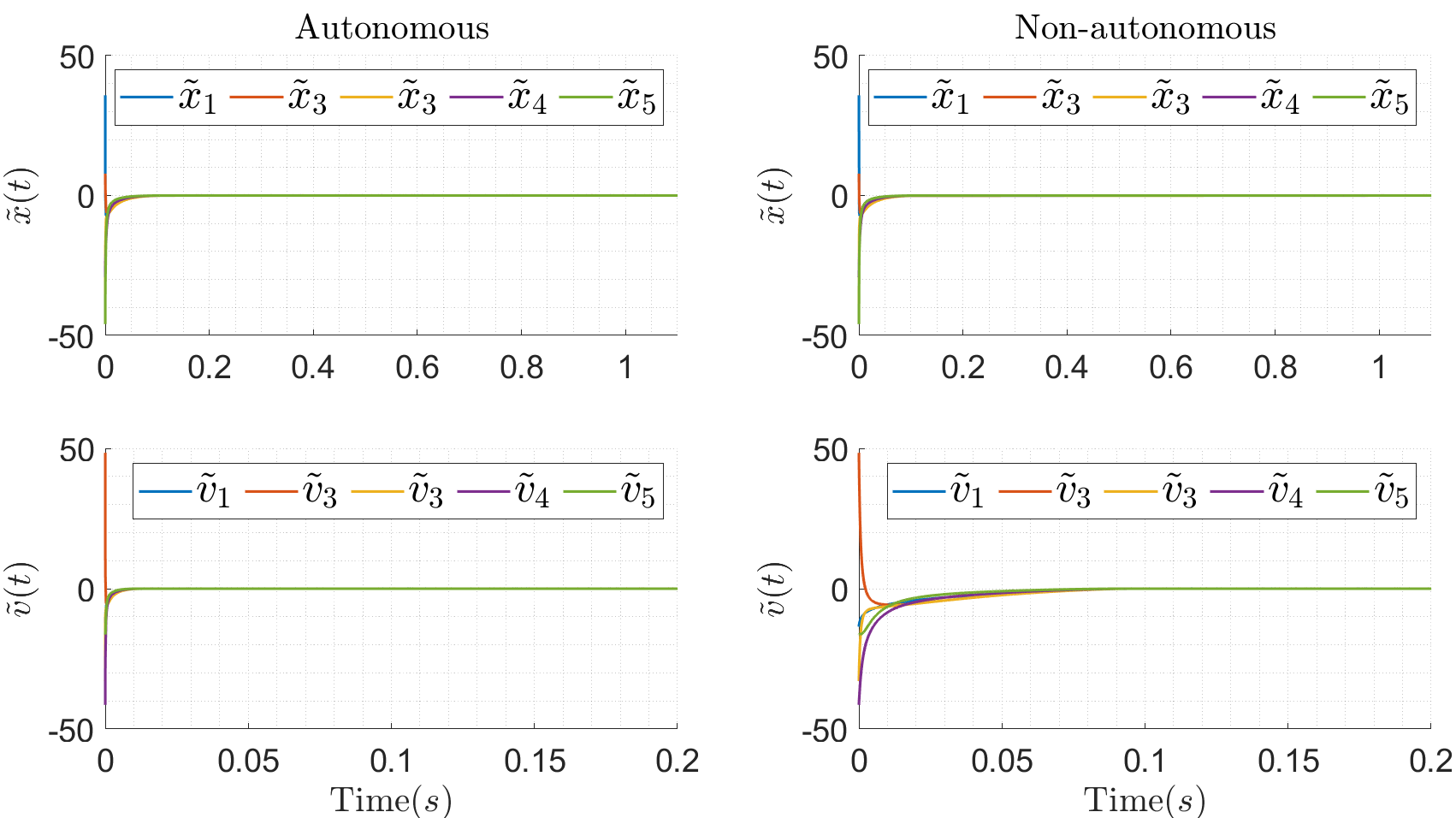}
 	\caption{Observation error of each agent.}
 	\label{fig:observer_errors}
 \end{figure}

\begin{figure}[h!]
 	\centering
 	\includegraphics[width=\textwidth]{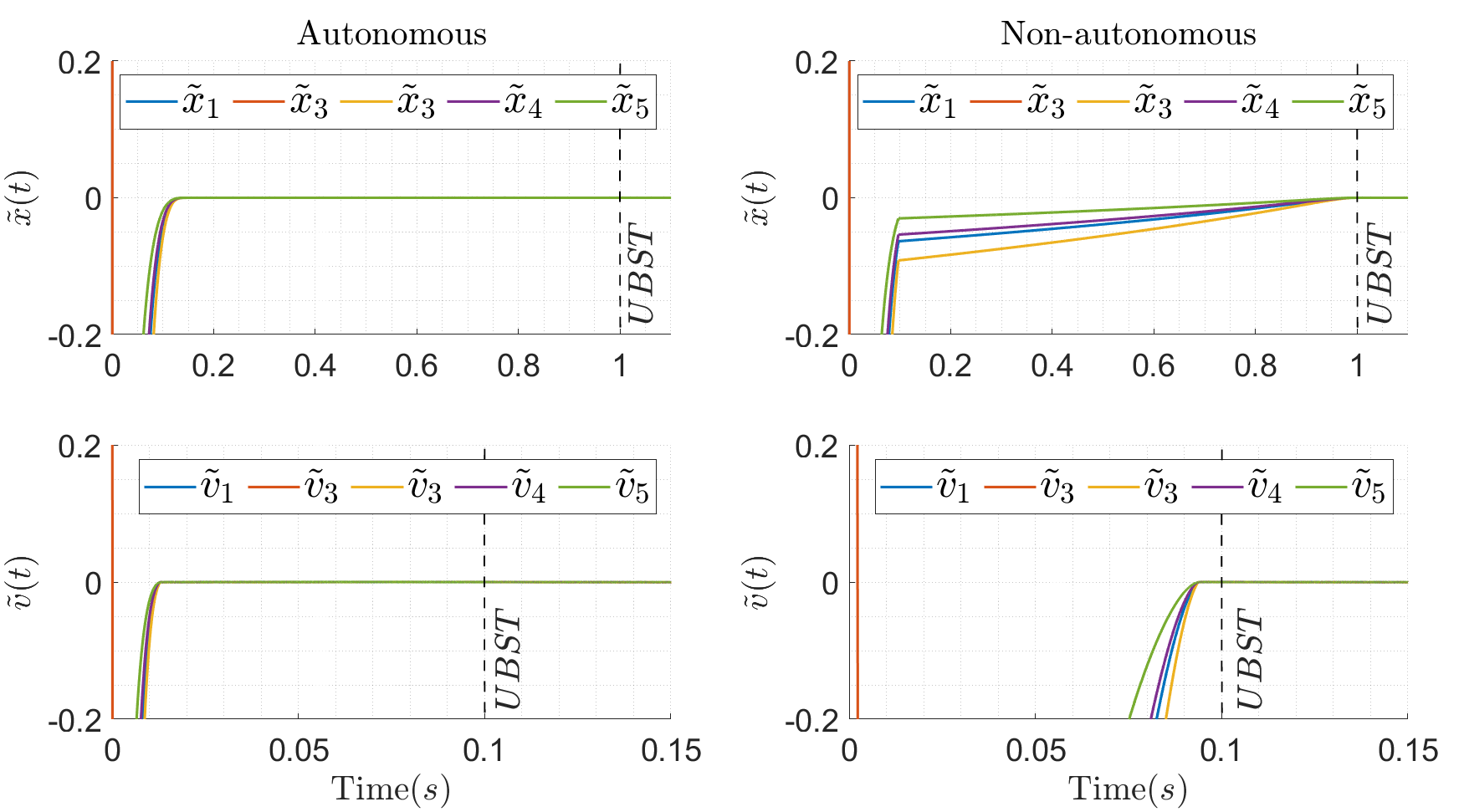}
 	\caption{Convergence of the observation error for each agent.}
 	\label{fig:observer_errors_zoom}
 \end{figure}
 
\begin{figure}[h!]
 	\centering
 	\includegraphics[width=\textwidth]{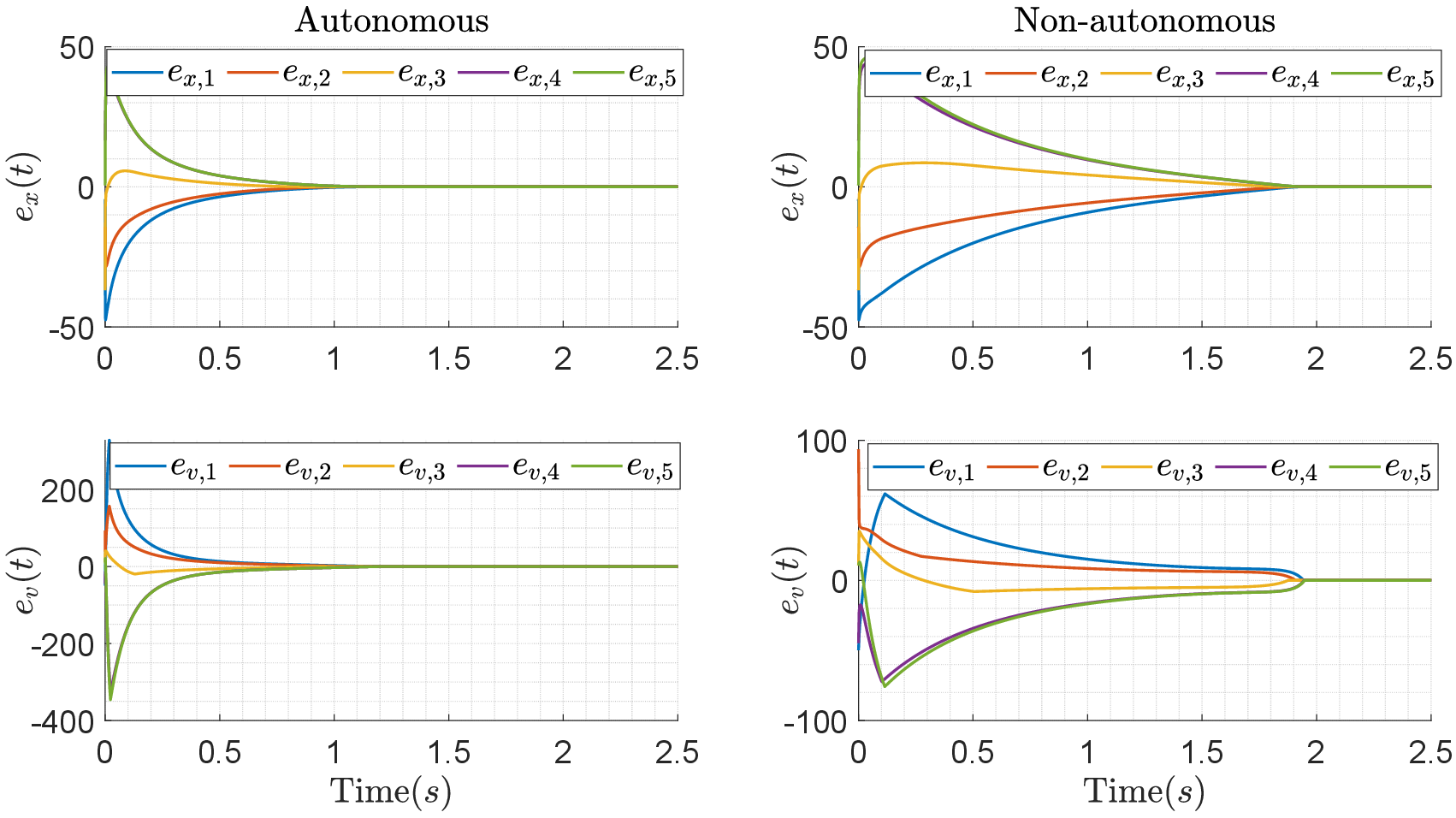}
 	\caption{Tracking error for each agent.}
 	\label{fig:tracking_errors}
 \end{figure}
 
 \begin{figure}[h!]
 	\centering
 	\includegraphics[width=\textwidth]{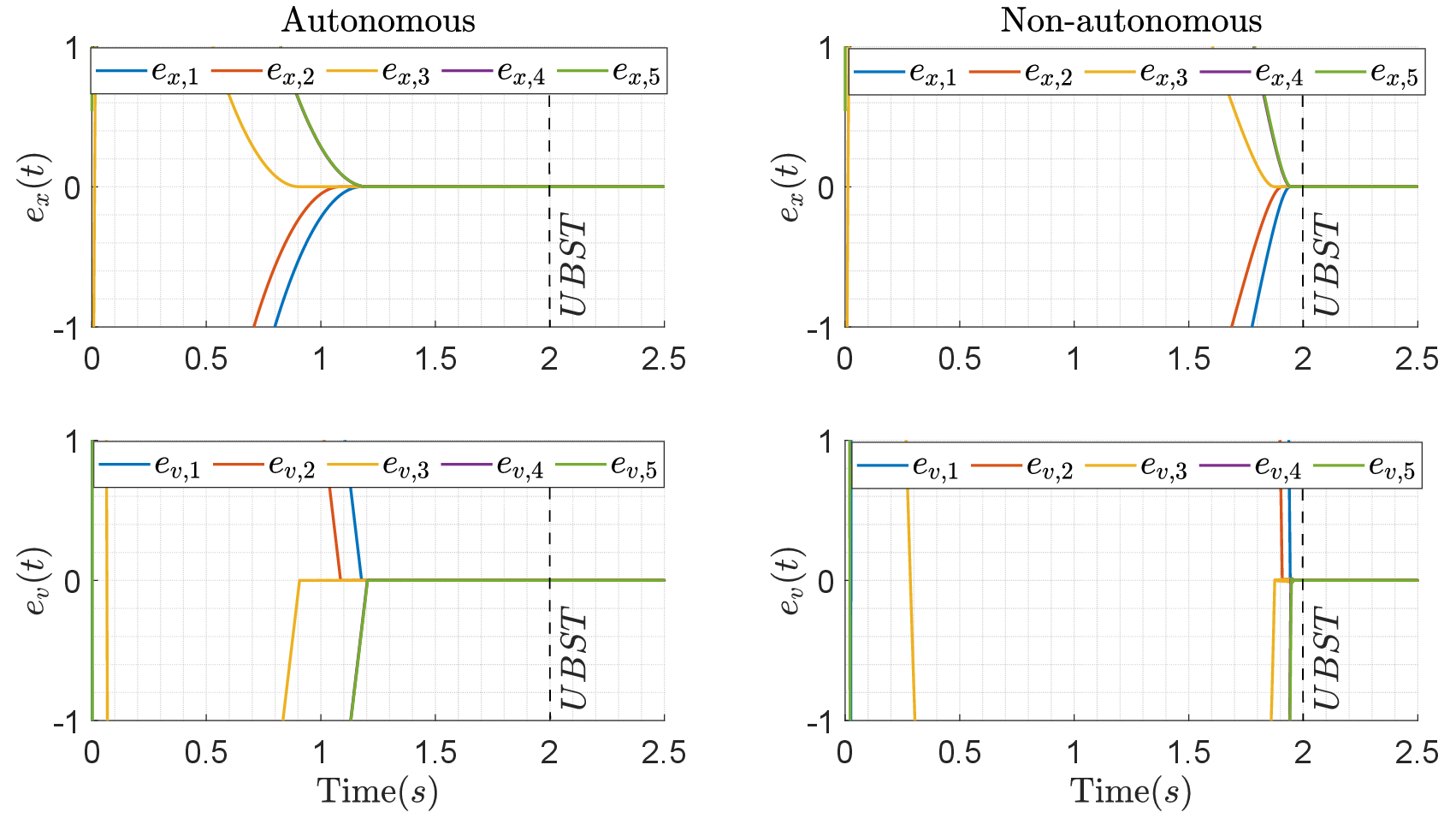}
 	\caption{Convergence of the tracking error for each agent.}
 	\label{fig:tracking_errors_zom}
 \end{figure}
\section{Conclusions}\label{sec:conclusions}
In this work, we presented novel protocols for the problem of consensus tracking with fixed-time convergence, for leader-follower multi-agent systems with double-integrator dynamics, where only a subset of followers has access to the state of the leader.
A distributed observer is proposed for each agent to estimate the leader state, and a local controller drives the agents towards the estimated state, both with fixed-time convergence. Two control strategies have been investigated and compared for the observer and controller parts. The first one is an autonomous protocol which ensures that the \textit{UBST} is set a priory by the user. Then, the previous strategy is redesigned using time-varying gains to obtain a non-autonomous protocol. This enables to obtain less conservative estimates of the \textit{UBST} while guaranteeing that the time-varying gains remain bounded. Future work includes the extension of the algorithm to chained form systems or high order MAS, the robustness against faults in the communication links and the extension of the protocol to directed graphs.

\section*{Apppendix}
\begin{lemma}\label{lem:sum_expo} \cite{Aldana-Lopez2018a}
Let $n\in N$. If $a=(a_{1},\dots,a_{n})$ is a sequence of positive numbers, then the following inequality is satisfied
\begin{equation}\label{eq:lemma_sum_expo}
    \frac{1}{n}\sum_{i=1}^{n}a_{i}\left(\alpha a_{i}^{p}+\beta a_{i}^{q} \right)^{k}\geq \left(\frac{1}{n}\sum_{i=1}^{n}a_{i} \right)\left(\alpha \left(\frac{1}{n}\sum_{i=1}^{n}a_{i} \right)^{p}+\beta \left(\frac{1}{n}\sum_{i=1}^{n}a_{i} \right)^{q}\right)^{k}.
\end{equation}{}
\end{lemma}

\begin{lemma}\cite{Basile1992} Let $z=[z_{1} \dots z_{n}]^{T}\in \mathbb{R}^{n}$ and
\begin{equation}
    \|z\|_{p}=\left(\sum_{i=1}^{n}\abs{z_{i}}^{p}\right)^{\frac{1}{p}}
\end{equation}
then, 
\begin{equation}
    \|z\|_{l}\leq \|z\|_{r}
\end{equation}{}
for $l> r>0$.
\label{lem:sum_abs}
\end{lemma}

\section*{Acknowledgements}
The authors would like to acknowledge the financial support of Intel Corporation for the development of this work.

\bibliographystyle{unsrt}

\end{document}